\pdfoutput=1
\documentclass[11pt]{article}

\usepackage[preprint]{acl}

\usepackage{amsmath, amssymb, amsthm}
\usepackage{tikz}
\usepackage{times}
\usepackage{latexsym}
\usepackage{amsmath, amssymb, amsthm}
\usepackage{amsmath,amsthm,amssymb}
\usepackage{mathtools}  
\usepackage{bm}  
\theoremstyle{definition}

\newtheorem{theorem}{Theorem}

\newtheorem{assumption}{Assumption}
\usepackage{xcolor}
\usepackage{mdframed}
\usepackage{marvosym}

\usepackage{fontawesome5}

\newcommand{\instructionresponse}[2]{%
    \begin{mdframed}[linecolor=gray!50,backgroundcolor=gray!10]
    \textbf{Instruction:} #1\\
    \textbf{Response:} #2
    \end{mdframed}
}

\usepackage[T1]{fontenc}

\usepackage[utf8]{inputenc}

\usepackage{microtype}

\usepackage{inconsolata}

\usepackage{graphicx}

\usepackage{tikz}
\usetikzlibrary{arrows, positioning}
\usepackage{wrapfig}
\usepackage{booktabs}
\usepackage{algorithm}
\usepackage{algpseudocode}
\usepackage{caption}
\usepackage{makecell}
\usepackage{multirow}
\usepackage{listings}
\usepackage{xcolor} 
\definecolor{lightergray}{rgb}{0.95, 0.95, 0.95}
\usepackage{tcolorbox}
\usepackage{wrapfig}

\usepackage{amsmath,amsfonts,bm}

\def\eqref#1{equation~\ref{#1}}

\def\1{\bm{1}}

\DeclareMathAlphabet{\mathsfit}{\encodingdefault}{\sfdefault}{m}{sl}
\SetMathAlphabet{\mathsfit}{bold}{\encodingdefault}{\sfdefault}{bx}{n}

\title{FinRipple: Aligning Large Language Models with Financial Market for Event Ripple Effect Awareness}

\author{
  \textbf{Yuanjian Xu}\textsuperscript{1}, 
  \textbf{Jianing Hao}\textsuperscript{1}, 
  \textbf{Kunsheng Tang}\textsuperscript{2}\thanks{This work was done while Kunsheng Tang was a research assistant at HKUST(GZ).}, 
  \textbf{Jingnan Chen}\textsuperscript{1}, \\
  \textbf{Anxian Liu}\textsuperscript{1}, 
  \textbf{Peng Liu}\textsuperscript{1}, 
  \textbf{Guang Zhang}\textsuperscript{1}\thanks{Corresponding author: \href{mailto:guangzhang@hkust-gz.edu.cn}{guangzhang@hkust-gz.edu.cn}} \\
  \textsuperscript{1}The Hong Kong University of Science and Technology (Guangzhou) \\
  \textsuperscript{2}University of Science and Technology of China
}

\tcbset{
    myboxstyle/.style={
        colback=gray!10,    
        colframe=gray!80,   
        boxrule=0.5mm,      
        arc=3mm,            
        left=1mm,           
        right=1mm,          
        top=1mm,            
        bottom=1mm,         
        fonttitle=\bfseries,
        coltitle=black,     
        colbacktitle=gray!20, 
        sharp corners,
    }
}

\begin{document}
\maketitle

\begin{abstract}

Financial markets exhibit complex dynamics where localized events trigger ripple effects across entities. Previous event studies, constrained by static single-company analyses and simplistic assumptions, fail to capture these ripple effects. While large language models (LLMs) offer emergent reasoning capabilities, their direct application falters due to structural market unawareness and limited capacity to analyze ripple effects. We propose FinRipple, an elegant framework that empowers LLMs with the ability to analyze ripple effects through financial theory-guided large-scale reinforcement learning. We begin by relaxing the assumptions of previous methods, incorporating a time-varying knowledge graph to accurately represent market structure. By seamlessly integrating classical asset pricing theory, we align the LLM with the market, enabling it to predict ripple effects. 
To the best of our knowledge, we are the first to provide a standardized definition of ripple effect prediction, a task that is extremely important yet unexplored in the financial domain.
Extensive experiments demonstrate that FinRipple provides a promising solution to this task.

\end{abstract}

\section{Introduction}

Financial markets are naturally complex, and sudden events can often impact the value of companies.~\citep{eventstudy_2017}.
A recent example underscores the impact of market reactions: On August 13, 2024, Starbucks announced Chipotle CEO Brian Niccol as its new CEO, triggering a 24.5\% surge in Starbucks' stock---the largest single-day gain in its history---while Chipotle's stock dropped over 10\%. The ripple effect extended to Starbucks' supply chain, with Jones Soda Co. rising 9.52\%, BRC Inc. gaining 6.25\%, and Celsius Holdings Inc. up 3.81\%. This example demonstrates the ripple effect that a single market event can have, not just on the company involved, but on other relevant companies~\citep{multisource_2023}.
Understanding and predicting these market ripple effects is crucial for informed financial decision-making, risk management, and strategic portfolio optimization.
Investors and risk managers rely on such insights into how company announcements~\citep{chief_2010, NPD_2015}, external news~\citep{new_role_2013, recall_2015}, or macroeconomic shocks~\citep{third_party_2012} may cascade through the market to anticipate broader impacts, enabling proactive strategies in volatile conditions~\citep{event_driven_prediction_2015,structured_event_prediction_2014}.
However, capturing these ripple effects remains a complex and underexplored challenge due to intricate, evolving, and interconnected factors at play.

\begin{figure*}
    \centering
    \includegraphics[width=0.98\textwidth]{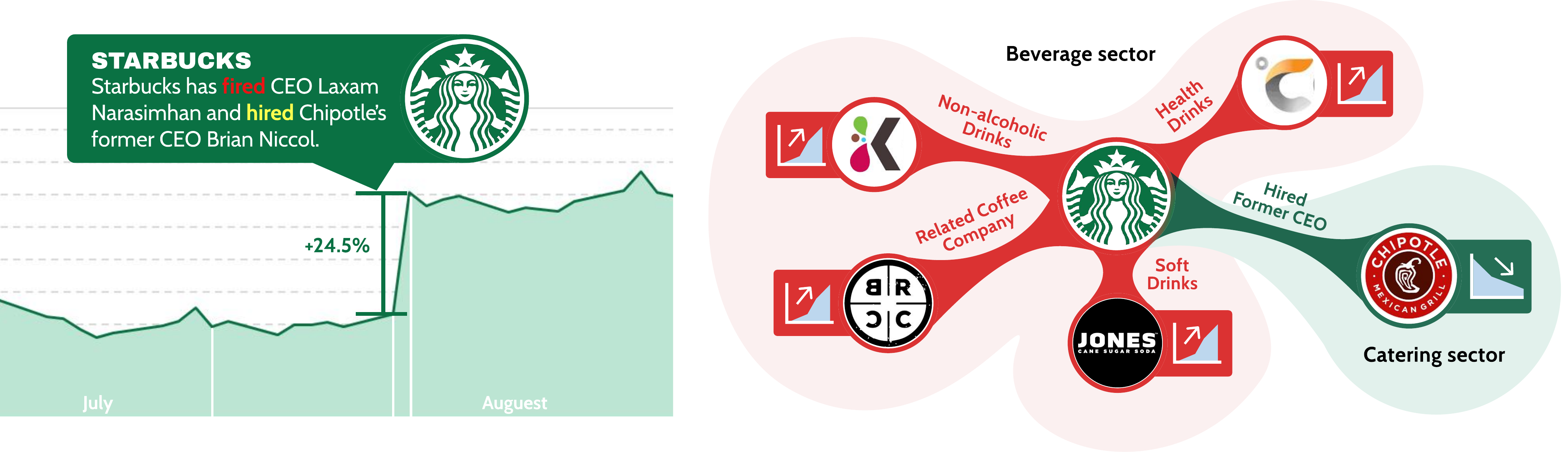}
    \vspace{-0.5em}
    \caption{An example of market ripple effects. The announcement of Starbucks's CEO change not only boosted its stock but also positively impacted other related companies in the beverage sector.}
    \label{fig:intro}
    \vspace{-0.7em}
\end{figure*}

Event studies have followed two main directions: case-by-case analysis and unified modeling based on learning theory. 
The former focuses on how specific market events affect the stock performance of a company or industry, which is a rather simplified assumption.
For example, \citet{austin_event_1993} analyzed patent innovations in biotechnology, \citet{lepetit_event_2004} studied M\&As in banking, and \citet{ramiah_event_2013} assessed stock reactions to green policies. 
While useful for direct impact assessment, these studies struggle to capture ripple effects across industries or the broader market.
On the other hand, learning-based approaches primarily use news sentiment to predict stock movements, acknowledging that a company's stock price is influenced by its related entities~\citep{news_survey_2023}. 
Recent innovations integrate multi-source information~\citep{multisource_2023} to enhance the prediction of emotions. \citet{finance_trans_2020} demonstrated that transformer-based models outperform lexicon-based and statistical approaches in event-driven word representation. 
However, relying solely on text sentiment can overlook critical dynamics—for instance, positive news for one company may negatively impact its associates. 
Thus, a more comprehensive framework is needed to model ever-changing market dynamics and explain complex intercompany relationships.

Recently, large language models (LLMs) have been widely used across various domains due to their advanced reasoning abilities~\citep{to_reasonig_2022}. 
They excel in structured information extraction~\citep{visltr_2024}, analogical reasoning~\citep{LLM_interpretab_reasoning_2022,cot_2022}, and question answering, making them promising candidates for analyzing event-driven ripple effects. 
Given their ability to model complex interactions, leveraging LLMs for financial market predictions is a natural step. 
However, financial markets, characterized by interconnected companies and dynamic relationships, evolve in response to various events, making the direct application of LLMs insufficient and potentially misleading~\citep{fts_fore_survey_2022, ADGAT_2021}.
To accurately model ripple effects, LLMs must be complemented with the latest market state.

A viable solution to address this challenge lies in integrating a time-varying financial knowledge graph (KG), which provides a structured view of the market by capturing up-to-date company relationships. 
Continuously updating the KG ensures a reliable snapshot of the evolving market~\citep{kg_fact_aware_2023}, enabling the modeling of dynamic corporate interactions~\citep{finKG_2020}. 
To effectively incorporate this knowledge into the LLM, we employ an adapter-based approach, injecting structured information without retraining the model from scratch. 
This method avoids potential information loss from retrieval-based methods and offers an extendable framework. 
By aligning LLMs with the financial market's unique characteristics, our approach significantly enhances their ability to analyze event-driven ripple effects. We validate its effectiveness in asset pricing and portfolio management through extensive experiments. The contributions of this work can be summarized as follows:
\begin{itemize}
\setlength{\itemsep}{0.2pt}
\setlength{\parskip}{0pt}
    \item FinRipple integrates classic asset pricing theory with advanced LLMs demonstrating strong performance in predicting excess returns while maintaining high interpretability.
    \item We rigorously validate our training framework and showcase its strong potential for real-world applications, such as asset pricing and portfolio management. Furthermore, detailed analyses illustrate the model's reasoning pathways, confirming its ability to provide reliable insights into the causal relationships driving ripple effects.
    \item We first formulate the under-explored ``ripple effect prediction'' task and provide an open-source benchmark, offering a unified evaluation standard for academia and industry.
\end{itemize}

\section{Related Work}

\subsection{Event studies in finance}
Event studies have been extensively employed to assess the impact of significant events on asset prices~\citep{eventstudy_2017}.
An event can be a firm announcement (e.g., the appointment of a new CMO) or an announcement made by competitors or regulatory bodies~\citep{event_2006}.
For example,~\citet{austin_event_1993} measured the innovative output of patents within the biotechnology industry;~\citet{lepetit_event_2004} discussed the effects of M\&As in the banking industry; and~\citet{ramiah_event_2013} analyzed the stock market reaction to green policy announcements. 
Due to simplistic assumptions, these methods often fail to capture the complexity and dynamics of modern financial markets.

Recognizing these limitations, researchers have explored unified modeling approaches based on learning theory, typically utilizing news sentiment analysis to predict stock price movements~\citep{trading_2010,sentimeng_2016}.
Recent advancements include the integration of multi-source information~\citep{multisource_2023}, the employment of more advanced embedding models~\citep{wordemb_dl_2019, finance_trans_2020}, and the usage of large language models (LLMs)~\citep{bloomberggpt,fingpt}.
Despite these promising developments, existing learning-based approaches struggle to fully capture the dynamic, time-varying relationships between companies and the evolving financial market.
Recent efforts on LLMs for financial tasks have aimed to overcome these challenges through multi-agent systems~\citep{fincon_2024,finmem_2024,stockagent_2024} and by infusing financial trading knowledge\citep{fin1_2024}.
Considering the structured, dynamic representations provided by knowledge graphs (KGs)~\citep{kg_dynamic_2023}, FinRipple takes an alternative approach by combining LLMs with financial KGs to capture ever-changing market dynamics and explain complex intercompany relationships.


\subsection{KG Augmented LLM}
Through the augmentation of knowledge graphs (KGs), existing methodologies strive to mitigate hallucinations, enhance reasoning capabilities, and facilitate the recall of specific facts~\citep{kg_multi_2024,cankg_survey_2023}.
Research on leveraging KGs to enhance LLMs can be broadly categorized into two main directions~\citep{mindmap_2023,cankg_survey_2023}: 1) integrating KGs during the pre-training phase, and 2) injecting KGs during the inference stage.
For methods that integrate KGs into LLM pre-training, the common practice involves designing knowledge-aware training objectives by either incorporating KG entities and relations into the training data~\citep{ernie_2019,ernie3_2021} or applying KG prediction tasks, such as link prediction, as additional supervision~\citep{graph_llm_pretrain_2022}.
These methods directly compress KG knowledge into the parameters of LLMs through supervision. 
However, constructing KGs containing trillions of words is challenging, and these approaches do not address the fundamental limitations of LLMs in terms of flexibility, reliability, and transparency.

Injecting structured symbolic knowledge from KGs into LLM inference aims to enhance contextual understanding, primarily by integrating knowledge at the input level.
Early efforts focused on fusing KG triples into the inputs of LLMs through attention mechanisms~\citep{kbert_2020,colake_2020} or by attaching graph encoders to LLM encoders to process KG information~\citep{graphencoder_2019}.
Subsequent work further adopted graph neural networks (GNNs) in parallel with LLMs for joint reasoning~\citep{QAGNN_2021}, as well as introducing interactions between text tokens and KG entities in the intermediate layers of LLMs~\citep{greaselm_2022,beyondcot_2023}.

\section{Methodology}
In this section, we commence by formalizing the mathematical framework for the ripple prediction task. We initially delineate the necessary inputs and outputs for the task, as well as the evaluation metrics. Subsequently, in Section~\ref{sec:pipeline}, we present the overall architecture of FinRipple. This architecture primarily comprises two pivotal components: knowledge injection and market alignment. The theoretical underpinnings of the optimization objectives can be found in Appendix~\ref{sec:problem_setting}.
\begin{figure*}
    \centering
    \includegraphics[width=0.98\linewidth]{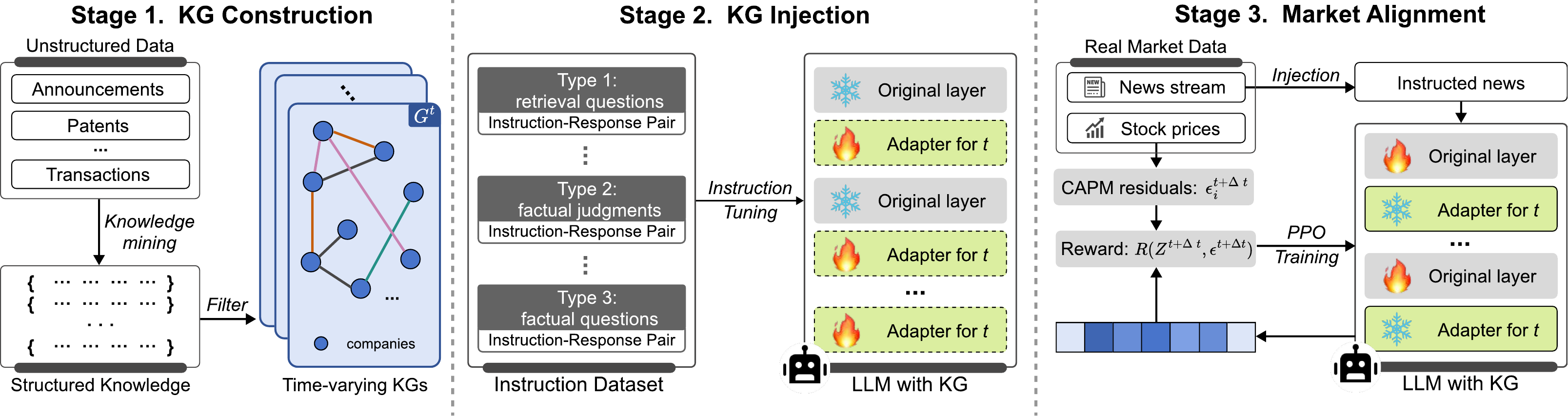}
    \vspace{-0.4em}
    \caption{Overview of FinRipple. The framework comprises three stages: (1) KG Construction: transforming unstructured data, such as announcements, patents, and transactions, into time-varying KGs that capture company relationships; (2) KG Injection: creating instruction datasets based on these KGs and using them to inject structured knowledge into adapters of an LLM without retraining the original layers; (3) Market Alignment: aligning predictions with real market reaction by using the correlation between the predicted event impact and CAPM residuals as the reward for PPO to optimize model performance. The adapter is frozen, and the analysis ability is parameterized into the original layers of the LLM.}
    \vspace{-0.4em}
    \label{fig:framework}
\end{figure*}

\subsection{Problem Formulation} 
\label{sec:problem_formulation}
The financial ecosystem evolves through the structured triad $\mathcal{M}_t = (\mathcal{C}_t, \mathcal{E}_t, \mu_t)$, where $\mathcal{C}_t$ captures the universe of public firms, $\mathcal{E}_t$ the event space, and $\mu_t$ the signed interaction measure. This measure's duality -- magnitude $|\mu_t(c_i,c_j)|$ for connection strength and polarity $\mathrm{sign}(\mu_t(c_i,c_j))$ for cooperation/competition -- synthesizes cross-channel dependencies spanning operational, financial, and strategic linkages.

Central to our framework is the propagator $\Phi_{e_t,\theta}$, a parametric operator that maps event-context pairs to forward shock distributions:
\[
\Phi_{e_t,\theta}: \underbrace{\mathcal{E}_t \times \mathcal{M}_t}_{\text{Event-Context}} \to \underbrace{\mathbb{R}^{\mathcal{C}_{t+\Delta t}}}_{\text{Shock Magnitudes}},
\]
where $\theta$ parameterizes network diffusion dynamics. Its validation requires grounding in asset pricing fundamentals: given stochastic discount factor $D_{t+\Delta t}$, the pricing error $\epsilon_j = \mathbb{E}_t[D_{t+\Delta t}R_j^{ex}]$ quantifies deviations from no-arbitrage equilibrium for firm $j$.

The core specification standardizes pricing errors by their cross-sectional volatility $\sigma_\epsilon = \sqrt{\mathrm{Var}(\epsilon_j)}$ to ensure scale invariance, yielding the propagator-constrained regression:  
\[
\frac{\epsilon_j}{\sigma_\epsilon} = \gamma_0 + \gamma_1 \Phi_{e_t,\theta}^j + \sum_{k=1}^K \Gamma_k X_{k,j} + \nu_j,
\]  
Here $\gamma_0$ captures baseline pricing anomalies, $\gamma_1$ quantifies the risk premium attributed to network-propagated shocks via the propagator component $\Phi_{e_t,\theta}^j$, and $\Gamma_k$ controls for $K$ standard risk factors $X_{k,j}$. The residual $\nu_j$ represents unexplained pricing noise with variance $\sigma_\nu^2$.  

Explanatory power is measured through normalized variance absorption:  
\[
R^2_\Phi = 1 - \frac{\mathbb{E}\left[ (\epsilon_j/\sigma_\epsilon - \hat{\gamma}_1\Phi^j)^2 \right]}{\mathrm{Var}(\epsilon_j/\sigma_\epsilon)},
\]  
where the expectation operator $\mathbb{E}[\cdot]$ averages over the cross-section of firms $\mathcal{C}_{t+\Delta t}$. Values $R^2_\Phi > 0.15$ indicate economically meaningful improvements over benchmark factor models. Statistical significance is evaluated through the robust $t$-statistic. A threshold $|t_{\gamma_1}| > 2.58$ ($p < 0.01$) establishes inference reliability.  More evaluation details can be found in Appendix~\ref{baselines}

\subsection{The pipeline of FinRipple}
\label{sec:pipeline}
As shown in Figure~\ref{fig:framework}, FinRipple starts with the construction of time-varying KGs that incorporate four relationships supported by prior research: leadership networks, mutual fund holdings, patent relationships, and supply chains. The specific data sources and construction process for the KG can be found in Appendix~\ref{appendix:Knowledge Graph Construction}. The next two key steps are KG injection and market alignment, which we will introduce in the following subsections.

\subsubsection{Knowledge Graph Injection}
\label{subsec:kg2instruction}

FinRipple implements time varying KG integration through structured instruction generation and parameter-efficient adaptation. Each dynamic KG snapshot $G^t = (\mathcal{C}_t, \mathcal{R}_t)$ contains the set of public firms $\mathcal{C}_t$ and time-sensitive relations $\mathcal{R}_t$, encoding four validated interaction types: leadership overlaps (CEO/board linkages), mutual fund cross-holdings, patent co-development relationships, and supply chain dependencies. These relations are projected into instructional text via templated transformations.

For each relational triple $(c_i, r_k, c_j) \in G^t$, the mapping operator $\mathcal{T}_k$ generates question-answer pairs that capture both qualitative and quantitative aspects of the relationship. A supply chain example produces:  
\instructionresponse{``Identify primary suppliers for $c_i$ in 2023Q2''}{``$c_j$ provided \$2.3M semiconductor components with 98\% on-time delivery''}

The instruction set $\mathcal{D}^t = \{(x_i^t, y_i^t)\}$ integrates three query modalities: entity retrieval probes (e.g., ``List firms sharing board members with $c_j$''), factual verification tasks (``Did $c_i$ acquire $c_k$ in 2021?''), and quantitative inference questions (``What percentage of $c_j$'s R\&D budget funds joint patents with $c_i$?''). 
Ablation studies confirm the necessity of this multimodal design (Table~\ref{tab:ablation_study} in Appendix~\ref{appendix:other_exp}).

We use lightweight adapter modules to parameterize time-varying KGs in which $\phi_t$ as additional parameters -- distinct from and operating in parallel to the frozen base LLM parameters $\psi$. These adapters constitute only 3.2\% of the total parameter count while enabling temporal adaptation. To maintain temporal coherence, high-impact instructional pairs from prior periods are retained in a rotating buffer, ensuring persistent interdependencies remain accessible. The complete implementation -- including temporal alignment protocols and adapter initialization -- is detailed in Appendix~\ref{tag:details}.

\subsubsection{Market Alignment}
Before the training process, for each news item, we retrieve the corresponding KG for the relevant time and inject it into the adapter, enabling the model to adapt to the time-varying market structure. Importantly, each time we fine-tune the backbone of the LLM, the adapter, which stores the information of the KG, is reinitialized and then kept frozen, ensuring compatibility between the updated backbone parameters and the dynamically injected knowledge. The adapter, once frozen, functions as a static feature extractor that represents market features at specific times. Meanwhile, the LLM backbone learns to make predictions consistent with the current market context.
 During the market alignment phase, FinRipple is primarily based on large-scale reinforcement learning. By carefully designing the feedback mechanism, we integrate the classic CAPM theory with alignment technologies, endowing the LLMs with the ability to analyze the ripple effect. The propagator $\Phi_{e_t,\theta}$'s predictions are validated through CAPM residual analysis. For company $c_j \in \mathcal{C}_{t+\Delta t}$, define:
\begin{align*}
\begin{cases}
E[R_j^{t+\Delta t}] = R_f + \beta_j(R_m^{t+\Delta t} - R_f) \\
\epsilon_j^{t+\Delta t} = R_j^{t+\Delta t} - E[R_j^{t+\Delta t}]
\end{cases}
\end{align*}
where $\beta_j = \frac{\text{Cov}(R_j, R_m)}{\text{Var}(R_m)}$ is estimated via OLS over rolling windows. The propagator's output $Y^{t+\Delta t} \in \mathbb{R}^{|\mathcal{C}_t| \times d}$ is aggregated to shock magnitudes:
\[
Z_j^{t+\Delta t} = \sum_{i=1}^{|\mathcal{C}_t|} \mu_t(c_i, c_j) \cdot Y_{ij}^{t+\Delta t}
\]
The alignment between predicted shocks $Z^{t+\Delta t}$ and observed residuals $\epsilon^{t+\Delta t}$ is quantified through:
\begin{equation*}
\mathcal{R}(Z, \epsilon) = \underbrace{\frac{Z \cdot \epsilon}{\|Z\|\|\epsilon\|}}_{\text{direction match}} + \lambda \underbrace{\frac{\sum_{j=1}^{|\mathcal{C}_{t+\Delta t}|} \min(|Z_j|, |\epsilon_j|)}{\|\epsilon\|_1}}_{\text{magnitude coverage}}
\label{eq:reward}
\end{equation*}

The first term of the above reward function measures how precisely the predicted impacts can explain the CAPM residuals, ensuring the model accurately learns the influence magnitude of specific events.
At the same time, the regularization controlled by the hyperparameter \(\lambda\) maximizes the recall rate to cover as many relevant impacts as possible. 
The role of the regularization term is to evaluate the extent to which $Z^{t+\Delta t}$ covers $\epsilon^{t+\Delta t}$ by comparing their values element by element (during training, $\Delta t$ is set to 1). 
More training details can be found in Appendix~\ref{tag:details}.

\subsubsection{FinRipple}
We collect the reward \( R^t \) to fine-tune the LLM backbone using Proximal Policy Optimization (PPO), while keeping the adapter layers frozen. The fine-tuning process follows the pipeline described below. First, we iterate through all available news articles. For each news item, we inject the KG corresponding to the specific month into the adapter. This allows the model to adapt to the time-varying market structure encoded within the KG. Importantly, every time we fine-tune the model, we utilize a newly initialized adapter to ensure that the updated LLM backbone parameters are compatible with the dynamic knowledge injected from the KG.

Once the KG is injected, we proceed with PPO fine-tuning for the LLM backbone. The frozen adapter serves as a static market feature at a certain time, while the LLM backbone learns to make predictions that align with the current market context reflected in the news and KG data.

\section{Experiment}

\subsection{Baselines and Evaluation Metrics}
In this subsection, we provide a brief introduction to the benchmarks and metrics for the asset pricing task only. For further details and information on downstream tasks related to portfolio management, please refer to Appendix~\ref{baselines}.

\paragraph{Datasets}
We selected 10,000 news articles about S\&P 500 companies from January 1, 2020, to June 30, 2022, as the test set, while approximately 110,000 articles from other years were used for training. Detailed statistics on the dataset about news and KGs can be found in Appendix~\ref{sec:news_stat}.

\paragraph{Baselines} We adopt several mainstream methods to demonstrate that FinRipple offers a powerful solution for this task. 
The baselines are primarily divided into two categories. The first category tests the analogical reasoning capabilities of foundational LLMs, showing that untrained LLMs lack the ability to analyze event impact effectively. 
The basic Retrieval-Augmented Generation (RAG)~\citep{lewis2020retrieval} approach utilizes an embedding model to retrieve relevant subgraph information from the KG, enabling LLMs to assess impacts based on this data. Zero-Shot Inference provides instructions to the model along with news and concatenated graph information. 
However, due to the limited window size of LLMs, some graph data may be incomplete. 
For companies specifically mentioned in the news, a two-hop subgraph is concatenated; otherwise, random graph information fills the LLM's input window. In-Context Learning (ICL)~\citep{brown2020language} builds upon the Zero-Shot approach by adding an example to aid the LLM in reasoning.
The second category primarily includes fine-tuned variations of FinRipple, both with and without market alignment. It emphasizes that even if the LLM effectively absorbs the graph information, without aligning with market dynamics, the model still lacks the ability to effectively analyze the impact of events. 
\paragraph{Evaluation metrics}
To evaluate the effectiveness of FinRipple in analyzing financial market shocks, we designed an evaluation framework focusing on three metrics: (1) explanatory power on the mean of the residuals, (2) explanatory power on the variance of the residuals, and (3) the refusal-to-answer rate. The residuals, derived from a CAPM regression of stock returns against market returns, represent the portion of returns unexplained by market factors. We use these residuals to assess whether predicted event impacts significantly explain return variance through regression analysis and ANOVA, with $p$-values indicating statistical significance. Additionally, the refusal-to-answer rate evaluates the robustness of LLMs in generating meaningful responses in complex, event-driven contexts.

\subsection{Main Results Analysis}

\begin{table*}[h]
\centering
\resizebox{\textwidth}{!}{%
  \begin{tabular}{l ccc ccc ccc ccc ccc}
  \toprule
  \multirow{2}{*}{\textbf{Model}} & \multicolumn{3}{c}{\textbf{RAG}} & \multicolumn{3}{c}{\textbf{Zero-Shot}} & \multicolumn{3}{c}{\textbf{ICL}} & \multicolumn{3}{c}{\textbf{FinRipple/w-o alignment}} & \multicolumn{3}{c}{\textbf{FinRipple}} \\
  \cmidrule(lr){2-4} \cmidrule(lr){5-7} \cmidrule(lr){8-10} \cmidrule(lr){11-13} \cmidrule(lr){14-16}
   & \textbf{Coef.} & \textbf{p-value} & \textbf{R$^2$} & \textbf{Coef.} & \textbf{p-value} & \textbf{R$^2$} & \textbf{Coef.} & \textbf{p-value} & \textbf{R$^2$} & \textbf{Coef.} & \textbf{p-value} & \textbf{R$^2$} & \textbf{Coef.} & \textbf{p-value} & \textbf{R$^2$} \\
  \midrule
  llama2-7b-chat & 0.012& 0.452 & 0.009& 0.031 & 0.601 & 0.012 & 0.042 & 0.503 & 0.018 & 0.047 & 0.510 & 0.019 & 0.150* & 0.030 & 0.083 \\
  llama2-13b-chat & 0.103 & 0.305 & 0.054 & 0.079 & 0.349 & 0.039 & 0.098 & 0.281 & 0.061 & 0.102 & 0.287 & 0.058 & 0.242** & 0.009 & 0.193 \\
  llama3-8b-instruct & 0.091 & 0.318 & 0.047 & 0.072 & 0.402 & 0.037 & 0.107 & 0.254 & 0.058 & 0.110 & 0.249 & 0.060 & 0.278** & 0.004 & 0.251 \\
  vicuna-7b-chat & 0.118 & 0.247 & 0.063 & 0.102 & 0.298 & 0.052 & 0.129 & 0.198 & 0.081& 0.125 & 0.205 & 0.074 & 0.330*** & 0.001 & 0.310 \\
  vicuna-13b-chat & 0.248* & 0.032 & 0.248 & 0.148 & 0.149 & 0.082 & 0.176 & 0.098 & 0.102 & 0.171* & 0.040 & 0.108 & 0.395*** & 0.000 & 0.340 \\
  Phi-3.5-mini-instruct & 0.082 & 0.395 & 0.032 & 0.065 & 0.498 & 0.019 & 0.094 & 0.347 & 0.052& 0.096 & 0.340 & 0.045 & 0.245** & 0.006 & 0.155 \\
  gemma-2-9b-it & 0.097 & 0.298 & 0.048 & 0.083 & 0.354 & 0.038 & 0.112 & 0.245 & 0.063 & 0.109 & 0.252 & 0.061 & 0.290*** & 0.001 & 0.215 \\
  \midrule
  GPT 3.5 & 0.083 & 0.398 & 0.028 & 0.062 & 0.051 & 0.075 & 0.056** & 0.004 & 0.112 & / & / & / & / & / & / \\
  GPT o1-preview & 0.152 & 0.342 & 0.047 & 0.119 & 0.392 & 0.056 & 0.192 & 0.229 & 0.082 & / & / & / & / & / & / \\
  GPT 4o-mini & 0.124 & 0.312 & 0.042 & 0.312* & 0.013 & 0.035 & 0.104 & 0.879 & 0.103 & / & / & / & / & / & / \\
  \bottomrule
  \end{tabular}
}
\vspace{-0.4em}
\caption{Comparison of baselines and FinRipple on LLMs. This table focuses on the explanatory power on the value of the CAPM residuals. The significance levels are indicated as follows: * \( p < 0.05 \), ** \( p < 0.01 \), *** \( p < 0.001 \). Note that cells containing a slash (/) indicate that the model does not have open-sourced weights available.}
\vspace{-0.4em}
\label{tab:capm_r_baseline_comparison}
\end{table*}

\begin{table*}[h]
\centering
\resizebox{\textwidth}{!}{%
  \begin{tabular}{l ccc ccc ccc ccc ccc}
  \toprule
  \multirow{2}{*}{\textbf{Model}} & \multicolumn{3}{c}{\textbf{RAG}} & \multicolumn{3}{c}{\textbf{Zero-Shot}} & \multicolumn{3}{c}{\textbf{ICL}} & \multicolumn{3}{c}{\textbf{FinRipple/w-o alignment}} & \multicolumn{3}{c}{\textbf{FinRipple}} \\
  \cmidrule(lr){2-4} \cmidrule(lr){5-7} \cmidrule(lr){8-10} \cmidrule(lr){11-13} \cmidrule(lr){14-16}
   & \textbf{ANOVA-F} & \textbf{ANOVA-p} & \textbf{ES} & \textbf{ANOVA-F} & \textbf{ANOVA-p} & \textbf{ES} & \textbf{ANOVA-F} & \textbf{ANOVA-p} & \textbf{ES} & \textbf{ANOVA-F} & \textbf{ANOVA-p} & \textbf{ES} & \textbf{ANOVA-F} & \textbf{ANOVA-p} & \textbf{ES} \\
  \midrule
  llama2-7b-chat & 1.624 & 0.231 & 0.089 & 1.304 & 0.274 & 0.068 & 2.392 & 0.097 & 0.108 & 2.565 & 0.082 & 0.092 & 3.123* & 0.033 & 0.142 \\
  llama2-13b-chat & 2.175 & 0.139 & 0.102 & 1.782 & 0.188 & 0.082 & 2.634 & 0.075 & 0.117 & 3.052* & 0.051 & 0.105 & 4.103** & 0.012 & 0.198 \\
  llama3-8b-instruct & 1.210 & 0.324 & 0.085 & 2.221 & 0.141 & 0.099 & 2.452 & 0.088 & 0.112 & 2.835 & 0.069 & 0.101 & 4.110** & 0.010 & 0.203 \\
  vicuna-7b-chat & 0.910 & 0.452 & 0.071 & 1.512 & 0.248 & 0.074 & 2.731 & 0.060 & 0.115 & 2.672 & 0.074 & 0.097 & 3.832* & 0.019 & 0.341 \\
  vicuna-13b-chat & 2.703 & 0.112 & 0.115 & 2.910* & 0.058 & 0.110 & 3.001* & 0.052 & 0.125 & 3.932** & 0.031 & 0.119 & 5.231*** & 0.003 & 0.287 \\
  Phi-3.5-mini-instruct & 1.563 & 0.257 & 0.097 & 2.334 & 0.126 & 0.104 & 2.815 & 0.062 & 0.118 & 3.014* & 0.048 & 0.110 & 4.315** & 0.009 & 0.215 \\
  gemma-2-9b-it & 2.443 & 0.128 & 0.109 & 1.905 & 0.172 & 0.091 & 2.447 & 0.089 & 0.095 & 3.122* & 0.039 & 0.108 & 4.012** & 0.014 & 0.159 \\
  \midrule
  GPT 3.5 & 1.375 & 0.301 & 0.090 & 1.645 & 0.223 & 0.088 & 2.087 & 0.129 & 0.105 & / & / & / & / & / & / \\
  GPT 4.0-preview & 0.812 & 0.443 & 0.067 & 2.112 & 0.145 & 0.100 & 2.372 & 0.098 & 0.117 & / & / & / & / & / & / \\
  GPT 4o-mini & 2.153 & 0.144 & 0.099 & 2.875* & 0.059 & 0.108 & 3.245 & 0.061 & 0.145 & / & / & / & / & / & / \\
  \bottomrule
  \end{tabular}
  }
\vspace{-0.4em}
\caption{Comparison of baselines and FinRipple on various models using ANOVA analysis. ANOVA-F represents the F-value from the ANOVA test, indicating the ratio of systematic to error variance. ANOVA-p represents the p-value for statistical significance, with * for \( p < 0.05 \), ** for \( p < 0.01 \), and *** for \( p < 0.001 \). Eta Squared (ES) represents the correlation ratio, which indicates the proportion of variance explained by the model. Cells with a slash (/) indicate that the model cannot be fine-tuned using FinRipple due to unavailable open-source weights.}
\vspace{-0.4em}
\label{tab:anova_comparison}
\end{table*}

As shown in Table~\ref{tab:capm_r_baseline_comparison}, both open-source and closed-source LLMs face significant challenges in analyzing the impact of financial market events without domain-specific training. The results establish three critical insights into LLMs' capabilities for financial ripple effect prediction. General-purpose architectures demonstrate systematic limitations in event-driven scenarios, with RAG methods showing performance instability due to deficient event-context extraction and ICL providing negligible improvements over zero-shot baselines. The observed R\textsuperscript{2} values below 0.25 across multiple model families confirm these fundamental constraints.

A hierarchical pattern emerges in knowledge-enhanced approaches. Basic market information infusion yields marginal gains, while domain-adapted implementations exhibit transformative improvements. The performance differential between baseline and fine-tuned configurations reveals that market dynamics internalization, not mere data injection, drives meaningful capability enhancement. Notably, model scale proves secondary to domain alignment, as evidenced by smaller architectures outperforming larger counterparts post-adaptation.

The demonstrated success of targeted domain adaptation over architectural size or general capabilities suggests that isomorphic mapping between knowledge systems and market mechanisms enables causal reasoning beyond native model capacities. This repositions domain-specific alignment as the critical pathway for developing professional-grade financial analytics systems.

\begin{table*}[h]
    \centering
    \resizebox{\textwidth}{!}{%
    \begin{tabular}{l c c c c}
        \toprule
        \textbf{Benchmark} & \textbf{Daily Return ($R_d \times 10^{-1}$)} & \textbf{Sharpe Ratio ($S_a$)} & \textbf{Maximum Drawdown (MDD)} & \textbf{Win Rate} \\
        \midrule
        Equal Weighting & 0.034 & 0.882 & -0.351 & 0.582 \\
        Volatility Weighting & 0.041 & 1.021 & -0.312 & 0.643 \\
        Markowitz Model & 0.029 & 0.954 & -0.292 & 0.613 \\
        Min-Variance Weighting & 0.028 & 0.821 & -0.401 & 0.552 \\
        \textbf{FinRipple} & 0.052 & 1.153 & -0.283 & 0.685 \\
        \bottomrule
    \end{tabular}
    }
    \vspace{-0.5em}
    \caption{Summary of backtest results for different portfolio management strategies on S\&P 500 constituent stocks (January 2020 to June 2022). Note that the daily return is presented with a factor of $10^{-1}$ for better readability.}
    \vspace{-0.5em}
    \label{tab:strategy_performance}
\end{table*}

\subsection{Portofolio Management}

To further demonstrate the effectiveness of FinRipple, we implement a simple intraday trading strategy based on the event impact prediction. The strategy selects stocks that exhibit the highest positive predicted event-driven impacts and creates a daily portfolio that rebalances at the end of each trading day.
Specifically, the steps are as follows:
\begin{enumerate}
    \item Each morning, based on the predicted impact results, we rank all stocks in our universe by the magnitude of their predicted impact.
    \item The top 10\% of stocks with the highest predicted positive impact are selected for a long position, while the bottom 10\% with the highest predicted negative impact are shorted.
    \item At the end of the day, the portfolio is rebalanced, and the next day's selection is based on new predictions.
\end{enumerate}
In accordance with previous portfolio management studies~\citep{xu2024plutus}, we selected benchmarks including Equal Weighting, Volatility Weighting, the Markowitz Model, and Min-Variance Weighting. Furthermore, we employed multiple evaluation metrics, such as daily return ($R_d$), sharpe ratio ($S_a$), and maximum drawdown (MDD), as presented in Table~\ref{tab:strategy_performance}. To prevent data contamination, the backtest period was set from January 2020 to June 2022, ensuring the result reliability. A detailed introduction to portfolio strategies and their evaluation can be obtained in Appendix~\ref{baselines}.

The results show that accurately predicting the range of impacts from financial market events can significantly mitigate portfolio risks. 
The strategy based on FinRipple outperforms benchmarks in key metrics, including daily return, Sharpe ratio, and maximum drawdown, achieving a daily return of $0.052$, a Sharpe ratio of $1.153$, and a maximum drawdown of $-0.283$. In contrast, strategies like Equal Weighting and Min-Variance Weighting exhibit higher maximum drawdowns, indicating greater vulnerability to market shocks without precise impact predictions. 
Overall, accurate event impact forecasting is crucial for enhancing risk control and improving investment outcomes.



\begin{figure*}[h]
    \centering
    \includegraphics[width=0.99\linewidth]{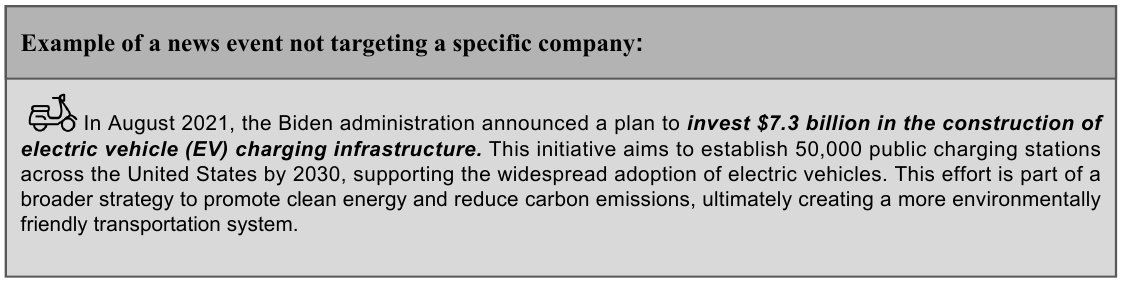}
    \vspace{-0.3em}
    \caption{An example where subgraph search is not applicable. As shown in the figure, this news event impacts the entire electric vehicle charging infrastructure industry rather than targeting a specific company.}
    \vspace{-0.5em}
    \label{fig:example}
\end{figure*}

\subsection{Other Analysis}

\paragraph{Knowledge Inject Analysis}
When effectively injecting KGs into LLMs, optimizing the model's understanding of market structures is paramount. One strategy involves using a preprocessing module to filter potential subgraphs as inputs. The simplest approach is to traverse one-hop and two-hop subgraphs related to a target company. While this method may be applicable in some contexts, it fails to capture the market's dynamic complexity, particularly in scenarios where events do not specifically target individual companies, such as those affecting entire supply chains.

Another strategy is to leverage RAG, which heavily relies on the performance of embedding models designed to recall companies that are ``semantically similar'' to specific queries. However, these embedding models often overlook the deeper market relationships associated with specific events when filtering for potentially impacted companies. This dependency can lead to significant misjudgments or biases in the model's event impact predictions. 

In contrast, the parameterization approach, which transforms KGs into adjustable parameters, provides a more comprehensive reflection of market trends and their complex interrelationships. 

\begin{figure}[h]
    \centering
    \includegraphics[width=0.95\columnwidth]{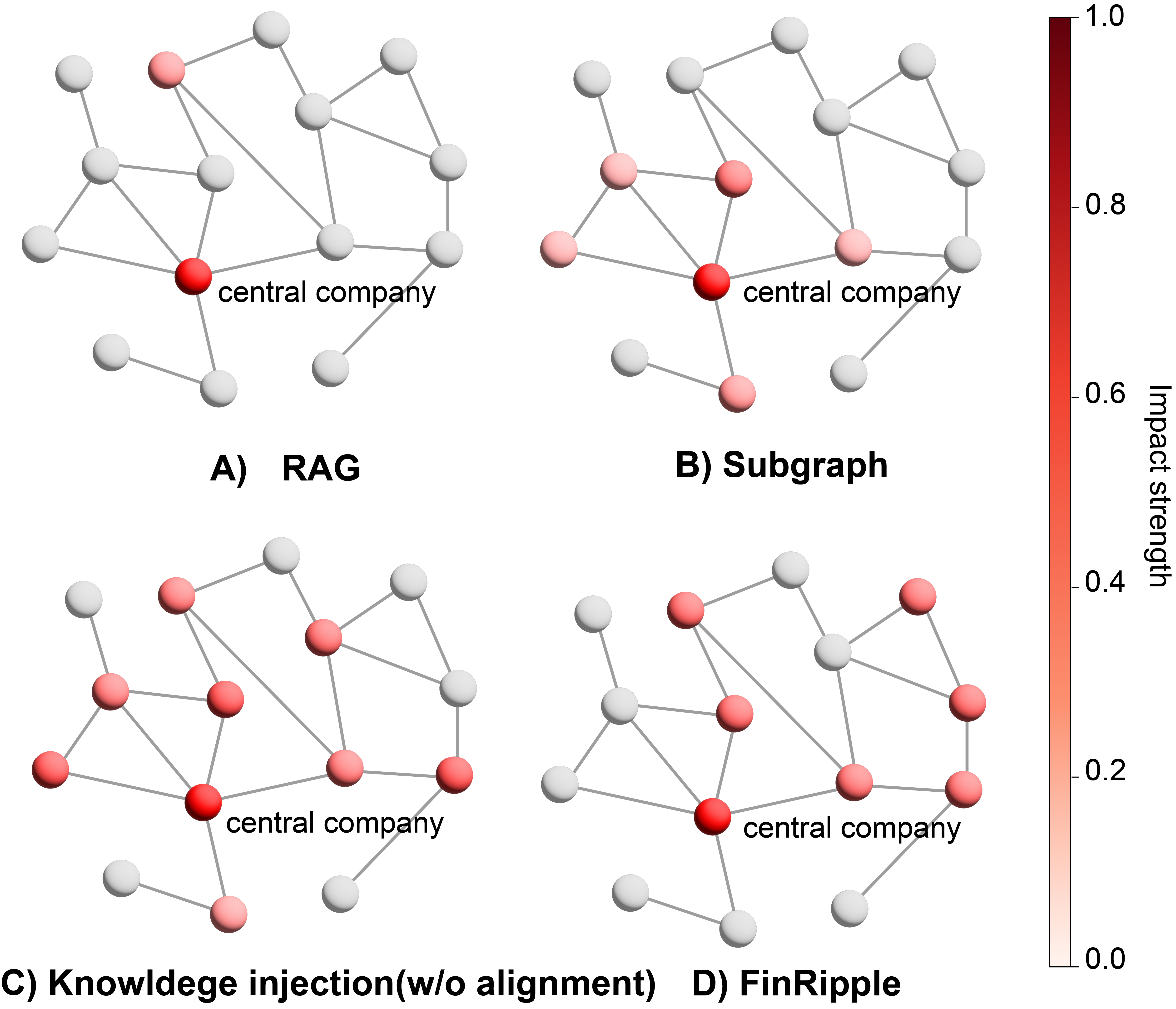}  
    \vspace{-0.2em}
    \caption{This diagram compares candidate companies identified by FinRipple and other methods. Due to the network's complexity, only selected nodes in the examples are shown for illustration purposes.}
    \vspace{-1em}
    \label{fig:wrapped}  
\end{figure}

This method enables dynamic adjustment and optimization of parameters during training, allowing the model to better capture the nonlinear dynamics of the market. By employing time-varying adapters, the model's adaptability to changes in market structure is enhanced, improving its responsiveness and predictive accuracy regarding market dynamic. For news events that focus on a specific central company, as Figure \ref{fig:wrapped} shows, RAG primarily retrieves based on semantic similarity, which often leads to a low recall rate when dealing with larger graphs. This limitation also affects first- and second-degree nodes, reducing the effectiveness of the retrieval process. Subgraph retrieval without alignment may select a larger number of relevant companies, but it often lacks the necessary logical structure to make meaningful predictions. FinRipple, by contrast, effectively captures not only the relationships among entities but also the logical pathways of impact from the central company, offering a more coherent and precise prediction of event impact. The clear propagation routes observed in FinRipple highlight its ability to model the cascading effects of an event through the network, accurately representing both direct and indirect influences. 

\begin{table}[h]
\centering
\resizebox{\columnwidth}{!}{%
\begin{tabular}{lccc}
\toprule
\textbf{Model} & \textbf{Zero-Shot} & \textbf{ICL} & \textbf{FinRipple} \\
\midrule
llama2-7b-chat & 0.41 $\pm$ 0.16 & 0.25 $\pm$ 0.09 & 0.21 $\pm$ 0.11 \\
llama2-13b-chat & 0.36 $\pm$ 0.18 & 0.13 $\pm$ 0.08 & 0.15 $\pm$ 0.09 \\
llama3-8b-instruct & 0.45 $\pm$ 0.19 & 0.11 $\pm$ 0.07 & 0.14 $\pm$ 0.08 \\
vicuna-7b-chat & 0.39 $\pm$ 0.14 & 0.22 $\pm$ 0.10 & 0.23 $\pm$ 0.05 \\
vicuna-13b-chat & 0.34 $\pm$ 0.15 & 0.13 $\pm$ 0.02 & 0.10 $\pm$ 0.04 \\
Phi-3.5-mini-instruct & 0.48 $\pm$ 0.21 & 0.31 $\pm$ 0.12 & 0.26 $\pm$ 0.09 \\
gemma-2-9b-it & 0.38 $\pm$ 0.17 & 0.23 $\pm$ 0.08 & 0.18 $\pm$ 0.06\\
\midrule
GPT 3.5 & 0.32 & 0.18  & / \\
GPT 4.0-preview & 0.14 & 0.10 & / \\
GPT 4o-mini & 0.12 & 0.09  & / \\
\bottomrule
\end{tabular}
}
\vspace{-0.5em}
\caption{Refusal-to-Answer Rate Comparison. The fluctuating values indicate variation under different temperature settings. This experiment is conducted on our benchmark, where refusal-to-answer samples are those that could not be post-processed into valid outputs.}
\vspace{-0.5em}
\label{tab:reject}
\end{table}

\begin{figure*}[h]
    \centering
    \includegraphics[width=0.96\textwidth]{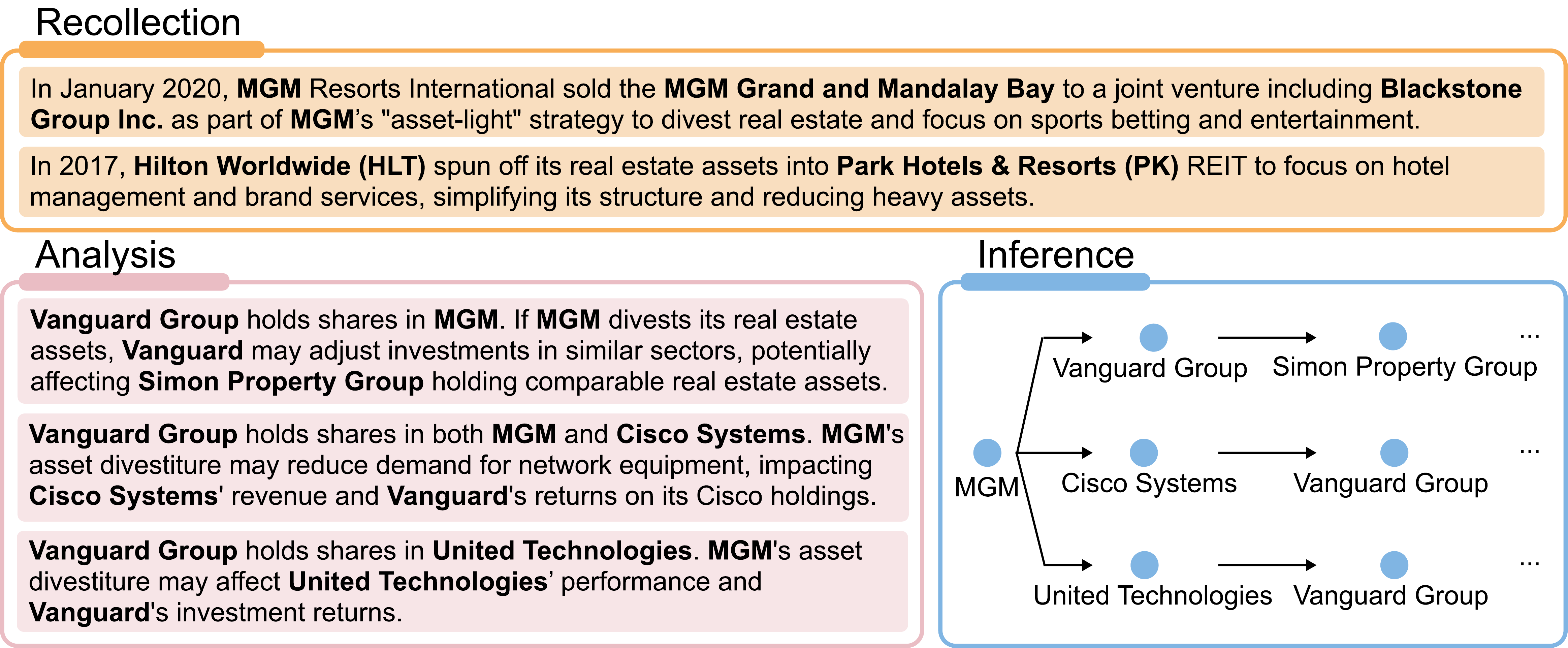}
    \vspace{-0.4em}
    \caption{Using CoT to analyze the reasoning process of vicuna-13b-chat. The model is aligned by FinRipple.}
    \vspace{-0.4em}
    \label{fig:case}
\end{figure*}

\paragraph{Refusal-to-Answer Rate Analysis:} In line with our experience, the refusal-to-answer rate largely depends on the model's instruction-following capability. As shown in Table~\ref{tab:reject}, zero-shot approaches exhibit systematically higher refusal rates with greater volatility across model architectures, reflecting fundamental limitations in interpreting complex domain-specific instructions. This pattern holds particularly for smaller open-source models, where instruction misinterpretation manifests as high variance in refusal behavior.

Closed-source architectures demonstrate superior instruction grounding, achieving sub-0.15 refusal rates through advanced comprehension capabilities. The FinRipple framework bridges this capability gap by transforming instruction semantics into market-dynamics-aware representations. Its effectiveness correlates with base model competency - stronger initial instruction following enables more precise financial alignment, as evidenced by order-of-magnitude improvements in compliant models.
\paragraph{Case study}
We believe that the logical reasoning capability of LLMs lies in their ability to establish connections with previously acquired knowledge or patterns. Therefore, in the inference process, we employ a straightforward Chain-of-Thought (CoT)~\citep{wei2022chain} approach to capture the intricate reasoning pathways, leading to the refined outcomes of FinRipple, as shown in Figure \ref{fig:case}. We can clearly observe that the inference process of the LLM, after being aligned with the financial market, is divided into three distinct steps: the first step involves establishing connections with past news, the second step focuses on analysis, and the third step derives the impact pathways. It is worth noting that not all news articles can directly establish connections with past knowledge. News that has undergone pre-training or supervised fine-tuning (SFT) is often more likely to be fully recalled and integrated into reasoning processes.

\section{Conclusion}
In conclusion, we present FinRipple, a novel training framework that empowers LLMs to analyze and predict the ripple effects of sudden events in financial markets.
By constructing a time-varying financial KG and integrating it into the LLM using adapters, we align the model with the dynamic market structure without retraining from scratch. Our rigorous validation showcases FinRipple's strong potential in real-world applications like asset pricing and portfolio construction.

\section{Limitations}
FinRipple also faces several limitations that warrant further attention. First, it relies heavily on high-quality, explicit KGs for effective performance. Without timely and high-quality data, the system’s capabilities may degrade over time, resulting in outdated knowledge and less relevant reasoning processes. At the same time, any delay in updating the KG may affect the relevance and accuracy of the generated outputs, especially during periods of rapid market fluctuations. Lastly, the scalability of LLMs when integrated with large-scale KGs remains a major concern~\cite{ibrahim2024survey}. As KGs grow in size, the computational burden of incorporating all relevant entities and relationships into LLMs increases substantially. Moving forward, efficient data management strategies and techniques such as model pruning will be essential to balance performance and computational cost effectively.

\section{Acknowledgements}
This work is supported by the Guangzhou-HKUST(GZ) Joint Funding Program (No. 2024A03J0630).
\bibliography{custom}

\appendix
\newpage
\appendix
\section{Theoretical Analysis}
\label{sec:problem_setting}
\subsection{Problem Setting}
Let $C = \{c_1, \ldots, c_n\}$ be a set of companies and $E^t = \{e_1^t, \ldots, e_m^t\}$ a set of events at time $t$. Given: True impact function: $f(c_i, e_j^t)$ for company $c_i$ and event $e_j^t$. Learnable model: $f_\theta(c_i, e_j^t)$ with parameters $\theta$.
We aim to minimize the empirical risk:
\[\label{eq:empirical_risk}
\hat{R}(\theta) = \frac{1}{n} \sum_{i=1}^n \biggl(\,\sum_{j=1}^m \bigl[f(c_i, e_j^t) - f_\theta(c_i, e_j^t)\bigr]\biggr)^{\!2},
\]
and bound the expected risk:
\[\label{eq:expected_risk}
R(\theta) = \mathbb{E}_{e \sim \mathcal{D}} \biggl[\,\frac{1}{n}\sum_{i=1}^n \bigl(f(c_i, e) - f_\theta(c_i, e)\bigr)^{\!2}\biggr].
\]

\begin{assumption}[Sparsity]\label{assump:sparsity}
For all $j\in[m]$ and $i\in[n]$:
\begin{align*}
&\text{(Event Sparsity)}  \:\bigl|\{i \mid f(c_i, e_j^t) \neq 0\}\bigr| \leq k, \\
&\text{(Company Sparsity)} \:\bigl|\{j \mid f(c_i, e_j^t) \neq 0\}\bigr| \leq l.
\end{align*}
\end{assumption}

\begin{assumption}[IID Sampling]\label{assump:iid}
The events $\{e_j^t\}_{j=1}^m$ are i.i.d. samples from the distribution $\mathcal{D}$.
\end{assumption}

\begin{assumption}[Non Dominant Error]
\label{ass:no_dominant_error}
For all \(i = 1,\dots,n\) and \(j = 1,\dots,m\), we have
\[
a_{ij} \le H \left(\sum_{j=1}^m a_{ij}\right),
\]
where \(H>0\) is a given constant.
\end{assumption}
\subsection{Generalization Bound}
\begin{theorem}[Generalization Bound]\label{thm:main}
Under Assumptions \ref{assump:sparsity}-\ref{ass:no_dominant_error}, if $\hat{R}(\theta) \leq \frac{B}{n}$ for some $B>0$, then
\[\label{eq:main_bound}
R(\theta) \leq \frac{B}{n} + C\frac{kl}{\sqrt{m}},
\]
where $C>0$ is a universal constant independent of $n,m,k,l$.
\end{theorem}

\begin{proof}
Define the per-instance error $a_{ij} := f(c_i, e_j^t) - f_\theta(c_i, e_j^t)$. The empirical risk becomes:
\begin{equation*}
\hat{R}(\theta) = \frac{1}{n}\sum_{i=1}^n \biggl(\,\sum_{j=1}^m a_{ij}\biggr)^{\!2} \leq \frac{B}{n}.
\end{equation*}
Combining Assumption~\ref{ass:no_dominant_error} and the above inequality, we have:
\[
\begin{aligned}
\sum_{i=1}^{n} \sum_{j=1}^{m} a_{ij}^2 &\leq H \sum_{i=1}^{n} \sum_{j=1}^{m} \left( a_{ij} \sum_{k=1}^{m} a_{ik} \right) \\
&= H \sum_{i=1}^{n} \left( \sum_{j=1}^{m} a_{ij} \right)^2 \leq HB.
\end{aligned}
\]
  


Define the loss class $\mathcal{L} = \{(e,c_i) \mapsto a_{ij}^2 \mid \theta \in \Theta\}$. By the Rademacher complexity bound~\cite{rademacher_2008,rademacher_2019}:
\begin{equation*}
\mathcal{R}_m(\mathcal{L}) \leq \sqrt{\frac{kl \log n}{m}}.
\end{equation*}

Applying standard generalization bounds with probability $1-\delta$, we obtain:
\begin{align*}
R(\theta) 
    &\leq \hat{R}(\theta) + 2\mathcal{R}_m(\mathcal{L}) + \sqrt{\frac{\log(1/\delta)}{2m}} \nonumber \\
    &\leq \frac{B}{n} + 2\sqrt{\frac{kl \log n}{m}} + \sqrt{\frac{\log(1/\delta)}{2m}} \nonumber \\
    &\leq \frac{B}{n} + C\frac{kl}{\sqrt{m}}, \label{eq:final_bound}
\end{align*}
where constant $C$ absorbs all logarithmic factors and numerical constants.
\end{proof}

\section{Datasets Details}\label{sec:news_stat}
Data preparation is critical in ensuring the quality and relevance of the input information for our model. This phase is bifurcated into two primary components: the collection of news events and the construction of the time-varying financial KG.

\subsection{News Collection and Processing: } 
The original 792,684 news articles are sourced from Dow Jones News Services and the Wall Street Journal, and stored as structured XML files.
The structured dataset comprises eight variables, including \{Publication\_datetime, Publisher\_name, Region\_code, Company\_code, Title, Body, Word\_count, Action\}.
Detailed descriptions of these variables are provided in Table \ref{tab:variable}.
Using the `Company\_code' variable, we filtered out 129,753 news articles about individual S\&P 500 firms, covering the period from March 8, 2001, to October 30, 2023.
After removing the irrelevant variables, the remaining eight variables and their descriptions are detailed in Table \ref{tab:variable}.
Figure \ref{fig:stat} (A) illustrates the distribution of news articles over time.
Notably, only 2 articles were recorded in 2001, while the highest number of articles, 16,103, was collected in 2012.
The analysis of word counts reveals that the average number of words per news article is 5,443.85, with the maximum word count reaching 77,086 and the minimum at 23 words.
This variation indicates a wide range of article lengths, from brief news briefs to extensive, in-depth reports.
Figure \ref{fig:stat} (B) presents the top ten companies with the highest number of news articles in the dataset.
This ranking highlights the companies that receive the most media attention, which may be attributed to their market influence, recent activities, or significant corporate actions.
We further analyzed the properties of daily news based on the `Action' variable, as shown in Figure \ref{fig:stat} (C).
63.94\% of the news articles pertain to organizational adjustments, which include changes in the company's business strategy, personnel, or departmental structures. 
36.06\% of the news articles involve new initiatives, such as the establishment of new companies, launching new projects or services, hiring new executives, and introducing new product lines, etc.

\begin{table*}[h]
    \centering
    \resizebox{1.5\columnwidth}{!}{ 
    \begin{tabular}{c|p{10cm}}
    \toprule[1.8pt]
       Variable  & Description \\
       \midrule[1.2pt]
       Publication\_datetime & Date and time of news article publication. It records the exact date and time when the news article was officially published.\\ \hline
       Publisher\_name & Name of the news publisher. It indicates the media outlet or organization that published the news article.\\ \hline
       Region\_code & Geographical region code. It specifies the geographic location relevant to the company's operational area.\\ \hline
       Company\_code & Unique identifier or code for the relevant company. A unique code that identifies the company mentioned in the news.\\  \hline
       Title & Title of news article. A brief headline that summarizes the main topic or event described in the news article.\\ \hline
       Body & The detailed news content.\\ \hline
       Word\_count & Number of total word count in the body of the news article.\\ \hline
       Action & Type of corporate action mentioned in the news. Its value can be `rep' or `add'.\\
    \bottomrule[1.8pt]
    \end{tabular}
    }
    \caption{The variables in the collected news articles dataset.}
    \label{tab:variable}
\end{table*}


\subsection{Knowledge Graph Construction: }\label{appendix:Knowledge Graph Construction}
\begin{figure*}[h]
    \centering
    \includegraphics[width=0.96\linewidth]{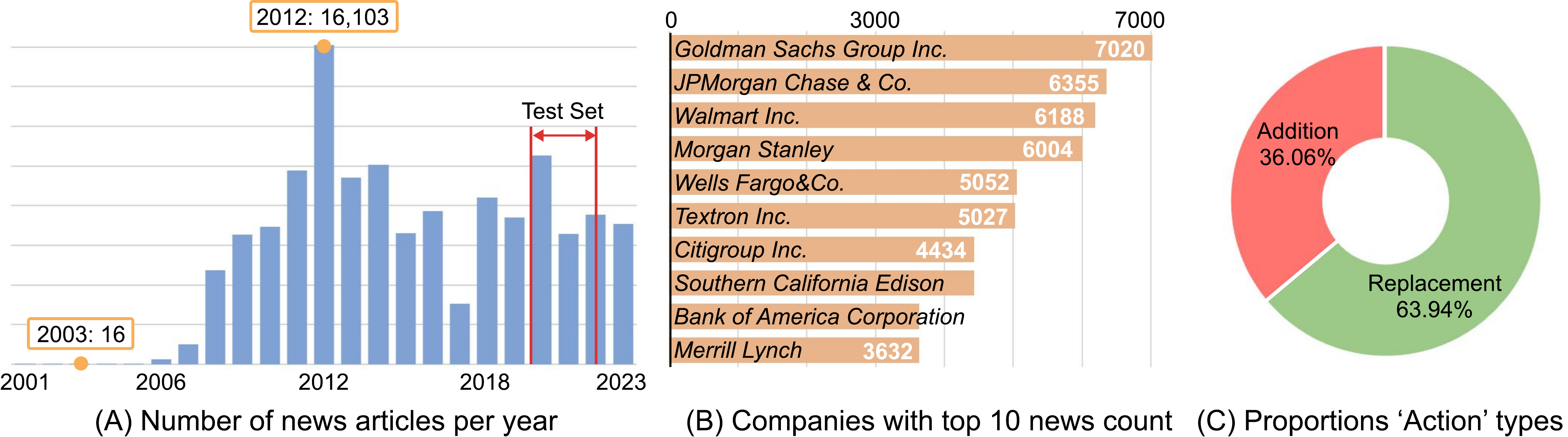}
    \caption{The statistics results of our collected news articles. (A) demonstrates the temporal distribution of news articles, (B) displays the company rankings with the top ten news counts, and (C) shows the properties of different corporate actions.}
    \label{fig:stat}
\end{figure*}

We constructed comprehensive financial KGs aimed at capturing the multifaceted interrelationships between companies and their potential impacts on profitability.
Each company is represented as a node, while the interrelationships between companies constitute the edges of the KGs.
To achieve this, we integrate various types of relationships derived from multiple data sources, ensuring a rich and nuanced representation of corporate interactions. 
\begin{itemize}
    \item \textbf{Technical Relevance Relationships.} We collect detailed and comprehensive information on firms' patents, including their corresponding Cooperative Patent Classification (CPC) codes, from the USPTO (United States Patent and Trademark Office) database to ensure a robust foundation for analyzing technical relevance and relationships between companies. Following the methodology outlined in~\cite{tech_2019}, we calculate pairwise technical closeness between two firms by measuring the correlation of CPC code distribution across their portfolios. An edge between two companies reflects their patent-based technical similarity. The strength of the edge is proportional to the degree of technical similarity, capturing the depth of their technological connections.

    \item \textbf{Supply Chain Relationships.} Information on firms' supply chains is extracted from the Compustat-Capital IQ database. Nodes represent companies, and edges indicate input-output relationships between companies. The strength of an edge is determined by the financial value of transactions between companies, providing a weighted representation of the intensity of their supply chain interactions.

    \item \textbf{Shared Leadership Relationships.} We obtain detailed information on firms' top leaders from the Boardex database. This data highlights interconnections between companies through shared executive affiliations. Edges denote the number of directors who simultaneously serve on the boards of two companies. This construction approach quantifies the degree of overlap in leadership structures, capturing the corporate governance ties between firms.

    \item \textbf{Mutual Fund Holding Relationships.} Data on mutual fund holdings of the listed U.S. firms is sourced from the Thomson/Refintiv database. Utilizing this information, we construct the holding-based relationships where an edge between two companies signifies that they are held by the same mutual fund. This relationship reflects the shared ownership structures and potential investment linkages among firms.
    
\end{itemize}
By extracting different types of relationships from these diverse data sources, we are able to construct a KG reflecting various dimensions of corporate interactions.
In the KG, each company and event is represented as a node, while the interrelationships between companies (such as collaborations or competitions) and the impact of events on companies constitute the edges of the graph.

In the process of constructing the KG, we pay special attention to associations supported by empirical financial research, such as future technology linkages evidenced by patent data and upstream-downstream enterprise relationships.
This focus ensures that the KG not only documents the static relationships but also delves deeply into how these relationships influence company performance under varying market conditions and in response to specific events.
The resulting KG  provides a comprehensive understanding of the interactions among S\&P 500 companies and offers the framework a robust and comprehensive understanding foundation.

Our KG dataset is divided into training and testing sets.
The training set covers the period from March 2001 to December 2019 (226 months), and the testing set encompasses the period from January 2020 to June 2022 (30 months).
Table \ref{tab:graph_stats} presents detailed statistics for both the training and testing KGs.
It includes the number of contained graphs, the average number of nodes per graph, the average number of edges per graph, and the distribution of relationship multiplicities between nodes.

\begin{table*}[h]
\centering
\resizebox{\textwidth}{!}{%
\begin{tabular}{lcccccc}
\toprule
& Graphs & \makecell{Avg. Nodes\\per Graph} & \makecell{Avg. Edges\\per Graph} & \makecell{Single\\Relationship (\%)} & \makecell{Dual\\Relationships (\%)} & \makecell{Triple\\Relationships (\%)} \\
\midrule
Training set & 226 & 6621.6018 & 13,844,186 & 92.7923 & 7.1956 & 0.0104 \\
Testing set & 30 & 6452.1667 & 14,228,088 & 95.0923 & 4.9007 & 0.0053 \\
\bottomrule
\end{tabular}
}
\caption{KG Data Statistics}
\label{tab:graph_stats}
\end{table*}

\section{FinRipple Details}
\label{tag:details}
\subsection{The detailed pipeline of FinRipple}
The training pipeline of FinRipple is detailed in Algorithm~\ref{alg:ppo_training_inference}.
\begin{algorithm*}[h]
\caption{Training Pipeline of FinRipple}
\label{alg:ppo_training_inference}
\begin{algorithmic}[1]
\Statex \hspace*{-\algorithmicindent} \textbf{Training Process:}
\Statex \hspace*{-\algorithmicindent} \textbf{Input:} KG s $G^t = \{G^1, \dots, G^n\}$, News data $N^t = \{N^1, \dots, N^m\}$, Pretrained LLM backbone $f_\theta$, Adapters $g_\phi$
\Statex \hspace*{-\algorithmicindent} \textbf{Output:} Updated LLM backbone parameters $\theta^\ast$
    \vspace{1.4mm}
    \hrule
    \vspace{1.4mm}
    \For{each time step $t$}
        \State Initialize an empty set $I=\{\}$, collect the KG  $G^t = \{C^t, R^t\}$ and news data $N^t = \{n_1^t, \dots, n_m^t\}$.
        \For{each article $n_j^t \in N^t$}
            \State Inject the corresponding KG  $G^t$ into the adapter $g_\phi$:
            \Statex $$g_\phi^{t} \leftarrow g_\phi(G^t),  f^{\phi}_\theta = g_\phi^{t} + f_\theta$$
            \State Inference the impact $Y_{ij}^{t+\Delta t}$ based on $n_j^t$:
            \Statex $$Y_{ij}^{t+\Delta t} \leftarrow f^{\phi}_\theta (n_j^t),  I \leftarrow I \cup Y_{ij}^t$$
            \State Compute the CAPM residuals:
            \Statex $$\epsilon_i^{t+\Delta t} = R_i^{t+\Delta t} - E(R_i^{t+\Delta t}),  E(R_i^{t+\Delta t}) = R_f + \beta_i (R_m^{t+\Delta t} - R_f)$$
            \State Calculate the reward at time $t$:
            \Statex $$R(Z^{t+\Delta t}, \epsilon^{t+\Delta t}) = \frac{Z^{t+\Delta t} \cdot \epsilon^{t+\Delta t}}{\|Z^{t+\Delta t}\| \|\epsilon^{t+\Delta t}\|} + \lambda \frac{\sum_i \min(Z^{t+\Delta t}_i, \epsilon^{t+\Delta t}_i)}{\|\epsilon^{t+\Delta t}\|_1} \quad \text{where } Z_j^{t+\Delta t} = \sum_{i=1}^{n} Y_{ij}^{t+\Delta t}$$
        \EndFor
        \State Update $\theta$ based on accumulated rewards.
\Statex
\[
\theta \leftarrow \theta + \alpha \mathbb{E}_t \left[ \nabla_\theta \log f^{\phi}_\theta(a_t | n_j^t) \frac{f^{\phi}_\theta(a_t | n_j^t)}{f^{\phi}_{\theta_{\text{old}}}(a_t | n_j^t)} \hat{A}_t \right] \quad \text{where } \hat{A}_t = R^t - V(n_j^t)
\]

    \EndFor
    \vspace{1.4mm}
    \hrule
    \vspace{1.4mm}
    \Statex \hspace*{-\algorithmicindent} \textbf{Inference Process:}
    \setcounter{ALG@line}{0}
    \Statex \hspace*{-\algorithmicindent} \textbf{Input:} new event $e_{new}$ and the corresponding KG  $G^{t_{new}}$.
    \vspace{1.4mm}
    \hrule
    \vspace{1.4mm}
    \State Inject $G^{t_{new}}$ into the frozen adapter $g_\phi$:
    \Statex $$g_\phi \leftarrow g_\phi(G^{t_{new}})$$
    \State Use the fine-tuned LLM backbone $f_{\theta^\ast}$ to predict the impact of the new event:
    \Statex $Y^t = f_{\theta^\ast}(G^{t_{new}}, e_{new})$
    where $Y^t$ represents the predicted impact of $e_{new}$ on the companies $C^t$.
    \State Output the predicted impact matrix $Y^t$.
\end{algorithmic}
\end{algorithm*}

\subsection{The Prompts used in FinRipple}
The following is a detailed prompt designed in FinRipple to guide the LLM for financial event analysis. The LLM is instructed to evaluate the impact of news on companies and provide a structured output. The news report will be placed in the ``[INSERT MARKET NEWS REPORT]'' section. The LLM is expected to determine the affected companies, classify the impact type, and assign an impact score from -10 to +10. A high positive or negative score indicates the strength of the potential market effect. The output should include specific company names, detailed descriptions and adhere strictly to the given format for consistency and clarity. An example is provided within the prompt to illustrate the expected response.

\noindent
\textbf{Instruction:}
\begin{lstlisting}
You are a financial event analyst focused on analyzing the potential impacts of news reports on the market. Based on the given news content and current market structure, evaluate and output the affected companies, the type of impact (positive, negative, or neutral), and a score representing the strength of the impact (ranging from -10 to +10, where -10 indicates a very negative impact, and +10 indicates a very positive impact). Provide specific company names and event descriptions for clarity and utility. Here is an example.
\end{lstlisting}

\noindent
\textbf{Input Example:}
\begin{lstlisting}
"Company A announces a partnership with Company B to jointly develop new technology, expected to significantly enhance production efficiency and increase market share."
\end{lstlisting}

\noindent
\textbf{Output Format Example:}
\begin{lstlisting}
{
  "impact_analysis": {
    "affected_companies": [
      {
        "name": "Company A",
        "impact_type": "positive",
        "impact_score": 8
      },
      {
        "name": "Company B",
        "impact_type": "positive",
        "impact_score": 6
      }
    ],
    "analysis": "The partnership between Company A and Company B is expected to enhance their technological capabilities and market competitiveness, likely increasing their revenues and stock prices. 
  }
}
\end{lstlisting}

\noindent
\textbf{Input (you need to analyze):}
\begin{lstlisting}
[INSERT MARKET NEWS REPORT]
\end{lstlisting}

\noindent
\textbf{Provide your result, strictly following the output format in the Example, without any additional output.}

\section{Asset Pricing Models}
Asset pricing models are essential tools in finance for understanding the relationship between risk and expected return. This appendix briefly introduces three prominent models: CAPM, Fama-French Three-Factor Model (Fama3), and Fama-French Five-Factor Model (Fama5).

\subsection{Capital Asset Pricing Model }

The CAPM describes the relationship between systematic risk and expected return. The expected return of an asset is proportional to its beta, which measures the sensitivity of the asset's returns to market returns. The formula for CAPM is:
\[
E(R_i) = R_f + \beta_i \left( E(R_m) - R_f \right),
\label{eq_a1}
\]
where \( E(R_i) \) represents the asset's expected return, \( R_f \) is the risk-free rate, \( \beta_i \) is the asset's beta that measures its sensitivity to market movements, and \( E(R_m) \) is the expected return of the market.

\subsection{Fama-French Three-Factor Model }

The Fama3 expands upon CAPM by including two additional factors: size and value. The size premium, denoted as Small Minus Big (SMB), captures the excess return of small-cap stocks over large-cap stocks, while the value premium, denoted as High Minus Low (HML), captures the excess return of high book-to-market stocks over low book-to-market stocks. The model is represented as:
\[
\begin{aligned}
E(R_i) = R_f &+ \beta_i \left( E(R_m) - R_f \right) \\ &+ s \times \text{SMB} + h \times \text{HML},
\label{eq_a2}
\end{aligned}
\]
where \( s \) and \( h \) represent the sensitivities of the asset's returns to the SMB and HML factors, respectively.

\subsection{Fama-French Five-Factor Model}

The Fama5 extends Fama3 by adding two more factors: profitability and investment. The profitability premium, denoted as Robust Minus Weak (RMW), captures the excess return of firms with high profitability over those with low profitability. The investment premium, denoted as Conservative Minus Aggressive (CMA), captures the excess return of firms with conservative investment policies over those with aggressive policies. The updated model is:
\[
\begin{aligned}
E(R_i) = R_f &+ \beta_i \left( E(R_m) - R_f \right) \\ &+ s \times \text{SMB} + h \times \text{HML} \\ &+ r \times \text{RMW} + c \times \text{CMA},
\label{eq_a3}
\end{aligned}
\]
where \( r \) and \( c \) represent the sensitivities to the RMW and CMA factors, respectively.

\subsection{Residuals and Market Anomalies}

Residuals of these models represent the portion of an asset's return not captured by the included risk factors. By analyzing residuals, investors can identify abnormal returns that the models fail to explain. These anomalies often arise due to market inefficiencies, information asymmetries, or other idiosyncratic risks not accounted for by the systematic factors in the models. Understanding residuals helps investors gain insights into potential mispricing and hidden variables in the market, revealing opportunities or risks that standard models overlook.
\begin{table}[htbp]
    \centering
    \resizebox{\columnwidth}{!}{%
    \begin{tabular}{l|c|c|c|c}
        \toprule
        Model & All & w/o RQ & w/o FJ & w/o FQ \\
        \midrule
        Gemma-2b-it & 84.6\% & 38.5\% & 15.4\% & 30.8\% \\
        Gemma-7b-it & 69.2\% & 30.8\% & 46.2\% & 46.2\% \\
        Llama-13b-chat & 61.5\% & 7.7\% & 15.4\% & 23.1\% \\
        \bottomrule
    \end{tabular}
    }
    \caption{Ablation study results for the three classes of questions: Retrieval Questions (RQ), Factual Judgments (FJ) and Factual Questions (FQ). The above results are averaged over five shuffles of the subgraph.}
    \label{tab:ablation_study}
\end{table}
\section{Other Experimental Results}\label{appendix:other_exp}

\subsection{The Accuracy of KG injection}
\label{appendix:Instruction Construction}
\begin{table*}[htbp]
\centering
\resizebox{1.8\columnwidth}{!}{  
\begin{tabular}{c|p{15cm}}
\toprule
    Problem Classification & Typical Questions \\
    \midrule
    \multirow{5}*{\makecell[c]{Retrieval\\Questions}} 
    & ``Which companies have a common CEO relationship with \{\}?'' \\
    & ``Which companies have an upstream-downstream relationship with \{\}?'' \\
    & ``Which companies have multiple relationships with \{\}?'' \\
    & ``Which companies have one relationship with \{\}?'' \\
    & ``Which companies have one relationship with \{\}?'' \\
    \midrule
    \multirow{3}*{\makecell[c]{Factual\\Judgments}} 
    & ``Are there supply chain upstream and downstream transactions between \{\} and \{\}?'' \\
    & ``Are the companies \{\} and \{\} held by the same fund?'' \\
    & ``Are the companies \{\} and \{\} held by the same fund?'' \\
    \midrule
    \multirow{3}*{\makecell[c]{Factual\\Questions}} 
    & ``What is the relationship between \{\} and \{\}?'' \\
    & ``What is the technical similarity between \{\} and \{\}?'' \\
    & ``What is the technical similarity score between \{\} and \{\}?'' \\
\bottomrule
\end{tabular}
}
\caption{The three classes of instruction questions generated from KGs.}
\label{tab:constrction_q}
\end{table*}

We used a random subgraph of 100 nodes for training, with an 8:2 split between the training and testing datasets. The results indicate that all three types of questions are beneficial. Note that some questions may not be answered correctly if the information needed is not fully covered by the training set. If all information is covered, our tests show that the adapter's memory accuracy reaches approximately 90\%.
We constructed three types of questions by traversing the KG, as shown in Table \ref{tab:constrction_q}. The first category, Retrieval Questions, focuses on identifying specific relationships between companies, such as shared CEOs or upstream-downstream connections. The second category, Factual Judgments, is used to determine whether certain relationships exist, such as common fund holdings or supply chain transactions. Finally, the third category, Factual Questions, aims to explore the details of relationships between entities, such as the nature of technical similarities or similarity scores.

\subsection{Evalidation on Other Asset Pricing Models}
In this subsection, we also evaluate FinRipple's ability to explain the residuals of other models including Fama3 and Fama5. Based on our experimental findings, as shown in Table~\ref{tab:fama3} and Table~\ref{tab:fama5}, we observe that the explanatory difficulty of Fama3 and Fama5 residuals gradually decreases. This reduction is primarily due to the stepwise exclusion of interfering factors from the residuals. The contributions of different variables were compared using standardized regression coefficients, as shown in Figure~\ref{fig:fama}. The results reveal that these factors exhibit distinct cyclical patterns. To account for these dynamics, we constructed training objectives based on the more challenging CAPM model. Although this approach increases the optimization difficulty, it ensures stable performance even when certain factors become less effective.
\subsection{Ablation Study on Graph Relationships}
We examined how different types of relationships in the graph affect prediction results. Removing any of the three relationship types—technical relevance, supply chain, or shared leadership—led to a drop in performance. As shown in Table~\ref{tab:rela}, among them, removing supply chain relationships caused the largest decline in explanatory power. The full model performed best and showed a statistically significant effect. This suggests that all three types of relationships are useful, with supply chain links being especially important.
\section{Baselines Details}
\label{baselines}
\subsection{Asset Pricing}
\subsubsection{Zero Shot}
Zero-shot inference enables the model to analyze a wider range of market scenarios without relying on specific examples. The prompt used is shown as follows:
\noindent
\textbf{Instruction:}
\begin{lstlisting}
You are a financial event analyst focused on analyzing the potential impacts of news reports on the market. Based on the given news content and current market structure, evaluate and output the affected companies (TICKER in SP500), the type of impact (positive, negative, or neutral), and a score representing the strength of the impact (ranging from -10 to +10, where -10 indicates a very negative impact and +10 indicates a very positive impact). Provide specific company names and event descriptions for clarity and utility. A market news report, company's knowledge graph information, specific requirements and output format will be provided below.
\end{lstlisting}
\noindent
\textbf{Market news report:}
\begin{lstlisting}
[INSERT MARKET NEWS REPORT]
\end{lstlisting}

\noindent
\textbf{Knowledge Graph (current market structure you can refer to):}
\begin{lstlisting}
[INSERT KNOWLEDGE GRAPH]
\end{lstlisting}

\noindent
\textbf{Requirement:}
\begin{lstlisting}
"Provide your result, strictly following the output format below, without any additional output."  
\end{lstlisting}

\noindent
\textbf{Output Format:}
\begin{lstlisting}
"Please provide your response in a structured JSON format. The JSON should have a top-level object with a single key 'impact_analysis'. The value of 'impact_analysis' should be an object containing two keys: 'affected_companies': An array of objects: 'name': The company's name (string) 'impact_type': The type of impact, e.g. 'positive' or 'negative' (string) 'impact_score': A numerical score representing the impact (integer) 'analysis': A string containing a brief analysis of the overall impact. Please ensure that the JSON is properly formatted and uses double quotes for strings.

Here's an example of how the structure should look:
{
    'impact_analysis': {
    'affected_companies': [
        {
        'name': 'Company Name',
        'impact_type': 'impact type',
        'impact_score': score
        },
        ...
    ],
    'analysis': 'Your analysis text here.'
    }
}"
\end{lstlisting}
\begin{table*}[h]
\centering
\resizebox{\textwidth}{!}{%
  \begin{tabular}{l ccc ccc ccc ccc ccc}
  \toprule
  \multirow{2}{*}{\textbf{Model}} & \multicolumn{3}{c}{\textbf{RAG}} & \multicolumn{3}{c}{\textbf{Zero-Shot}} & \multicolumn{3}{c}{\textbf{ICL}} & \multicolumn{3}{c}{\textbf{FinRipple/w-o alignment}} & \multicolumn{3}{c}{FinRipple} \\
  \cmidrule(lr){2-4} \cmidrule(lr){5-7} \cmidrule(lr){8-10} \cmidrule(lr){11-13} \cmidrule(lr){14-16}
   & \textbf{Coef.} & \textbf{p-value} & \textbf{R$^2$} & \textbf{Coef.} & \textbf{p-value} & \textbf{R$^2$} & \textbf{Coef.} & \textbf{p-value} & \textbf{R$^2$} & \textbf{Coef.} & \textbf{p-value} & \textbf{R$^2$} & \textbf{Coef.} & \textbf{p-value} & \textbf{R$^2$} \\
  \midrule
  Llama2-7b-chat & 0.021& 0.482 & 0.013 & 0.040 & 0.657 & 0.021 & 0.058& 0.287 & 0.145 & 0.090 & 0.520 & 0.152 & 0.310* & 0.021 & 0.275 \\
  Llama2-13b-chat & 0.132 & 0.405 & 0.074 & 0.095 & 0.445 & 0.065 & 0.158& 0.245 & 0.138 & 0.182& 0.314 & 0.195 & 0.445*& 0.013 & 0.390 \\
  Llama3-8b-instruct & 0.102 & 0.365 & 0.051 & 0.067 & 0.380 & 0.030 & 0.088 & 0.370 & 0.099 & 0.211 & 0.402 & 0.178 & 0.370& 0.007 & 0.400 \\
  vicuna-7b-chat & 0.158 & 0.235 & 0.095 & 0.112 & 0.400 & 0.078 & 0.215& 0.142& 0.134 & 0.250 & 0.188 & 0.256 & 0.515*** & 0.001 & 0.485 \\
  vicuna-13b-chat & 0.505** & 0.028* & 0.145 & 0.172& 0.210 & 0.123 & 0.290*& 0.031 & 0.255 & 0.365& 0.175 & 0.342 & 0.610*** & 0.001 & 0.550 \\
  Phi-3.5-mini-instruct & 0.097 & 0.512 & 0.032 & 0.056 & 0.670 & 0.026 & 0.075 & 0.470 & 0.086 & 0.153& 0.395 & 0.202 & 0.285** & 0.005 & 0.335 \\
  Gemma-2-9b-it & 0.112 & 0.298 & 0.061 & 0.089 & 0.423 & 0.047 & 0.178& 0.285 & 0.144 & 0.265 & 0.305 & 0.330 & 0.395*** & 0.001 & 0.445 \\
  \midrule
  GPT 3.5 & 0.060 & 0.455 & 0.018 & 0.045 & 0.550 & 0.039 & 0.069*& 0.018 & 0.106 & / & / & / & / & / & / \\
  GPT 4.0-preview & 0.165 & 0.328 & 0.045 & 0.119 & 0.389 & 0.063 & 0.195& 0.512& 0.138 & / & / & / & / & / & / \\
  GPT 4o-mini & 0.198& 0.215 & 0.051 & 0.145 & 0.312 & 0.055 & 0.155& 0.209& 0.121 & / & / & / & / & / & / \\
  \bottomrule
  \end{tabular}
}
\caption{Differences in the explanatory power of Fama3 residuals by baselines and FinRipple applied to LLMs. Significance levels: * \( p < 0.05 \), ** \( p < 0.01 \), *** \( p < 0.001 \). Cells with `/` indicate unavailable model parameters.}
\label{tab:fama3}
\end{table*}

\begin{table*}[h]
\centering
\resizebox{\textwidth}{!}{%
  \begin{tabular}{l ccc ccc ccc ccc ccc}
  \toprule
  \multirow{2}{*}{\textbf{Model}} & \multicolumn{3}{c}{\textbf{RAG}} & \multicolumn{3}{c}{\textbf{Zero-Shot}} & \multicolumn{3}{c}{\textbf{ICL}} & \multicolumn{3}{c}{\textbf{FinRipple/w-o alignment}} & \multicolumn{3}{c}{FinRipple} \\
  \cmidrule(lr){2-4} \cmidrule(lr){5-7} \cmidrule(lr){8-10} \cmidrule(lr){11-13} \cmidrule(lr){14-16}
   & \textbf{Coef.} & \textbf{p-value} & \textbf{R$^2$} & \textbf{Coef.} & \textbf{p-value} & \textbf{R$^2$} & \textbf{Coef.} & \textbf{p-value} & \textbf{R$^2$} & \textbf{Coef.} & \textbf{p-value} & \textbf{R$^2$} & \textbf{Coef.} & \textbf{p-value} & \textbf{R$^2$} \\
  \midrule
  Llama2-7b-chat & 0.018 & 0.489 & 0.014 & 0.042 & 0.670 & 0.025 & 0.078 & 0.260 & 0.152 & 0.127 & 0.445& 0.185 & 0.345** & 0.007 & 0.300 \\
  Llama2-13b-chat & 0.155* & 0.039 & 0.082 & 0.091 & 0.435 & 0.068 & 0.180& 0.428& 0.150 & 0.225& 0.309& 0.220 & 0.500*** & 0.001 & 0.420 \\
  Llama3-8b-instruct & 0.112 & 0.368 & 0.059 & 0.075 & 0.385 & 0.034 & 0.103 & 0.330 & 0.109 & 0.265& 0.306& 0.205 & 0.405*** & 0.001& 0.440 \\
  vicuna-7b-chat & 0.170* & 0.021 & 0.101 & 0.125 & 0.370 & 0.087 & 0.250& 0.303& 0.145 & 0.288& 0.107& 0.280 & 0.565*** & 0.001& 0.525 \\
  vicuna-13b-chat & 0.540** & 0.010 & 0.160 & 0.190* & 0.042 & 0.148 & 0.320& 0.315& 0.260 & 0.420& 0.111& 0.375 & 0.655*** & 0.000& 0.590 \\
  Phi-3.5-mini-instruct & 0.105 & 0.495 & 0.038 & 0.050 & 0.690 & 0.032 & 0.090 & 0.460 & 0.095 & 0.185& 0.422& 0.230 & 0.330** & 0.004 & 0.360 \\
  Gemma-2-9b-it & 0.140* & 0.028 & 0.068 & 0.087 & 0.425 & 0.048 & 0.205& 0.727& 0.155 & 0.305& 0.267& 0.360 & 0.430*** & 0.001& 0.485 \\
  \midrule
  GPT 3.5 & 0.070 & 0.435 & 0.023 & 0.038 & 0.585 & 0.039 & 0.085& 0.322& 0.120 & / & / & / & / & / & / \\
  GPT 4.0-preview & 0.180* & 0.031 & 0.050 & 0.125 & 0.390 & 0.062 & 0.220& 0.606& 0.150 & / & / & / & / & / & / \\
  GPT 4o-mini & 0.205& 0.629& 0.058 & 0.145 & 0.315 & 0.061 & 0.175& 0.703& 0.135 & / & / & / & / & / & / \\
  \bottomrule
  \end{tabular}
}
\caption{Differences in the explanatory power of Fama5 residuals by baselines and FinRipple applied to LLMs. Significance levels: * \( p < 0.05 \), ** \( p < 0.01 \), *** \( p < 0.001 \). Cells with `/` indicate unavailable model parameters.}
\label{tab:fama5}
\end{table*}

\begin{figure*}[h]
    \centering
    \includegraphics[trim=0 3 0 5, clip,width=0.98\linewidth]{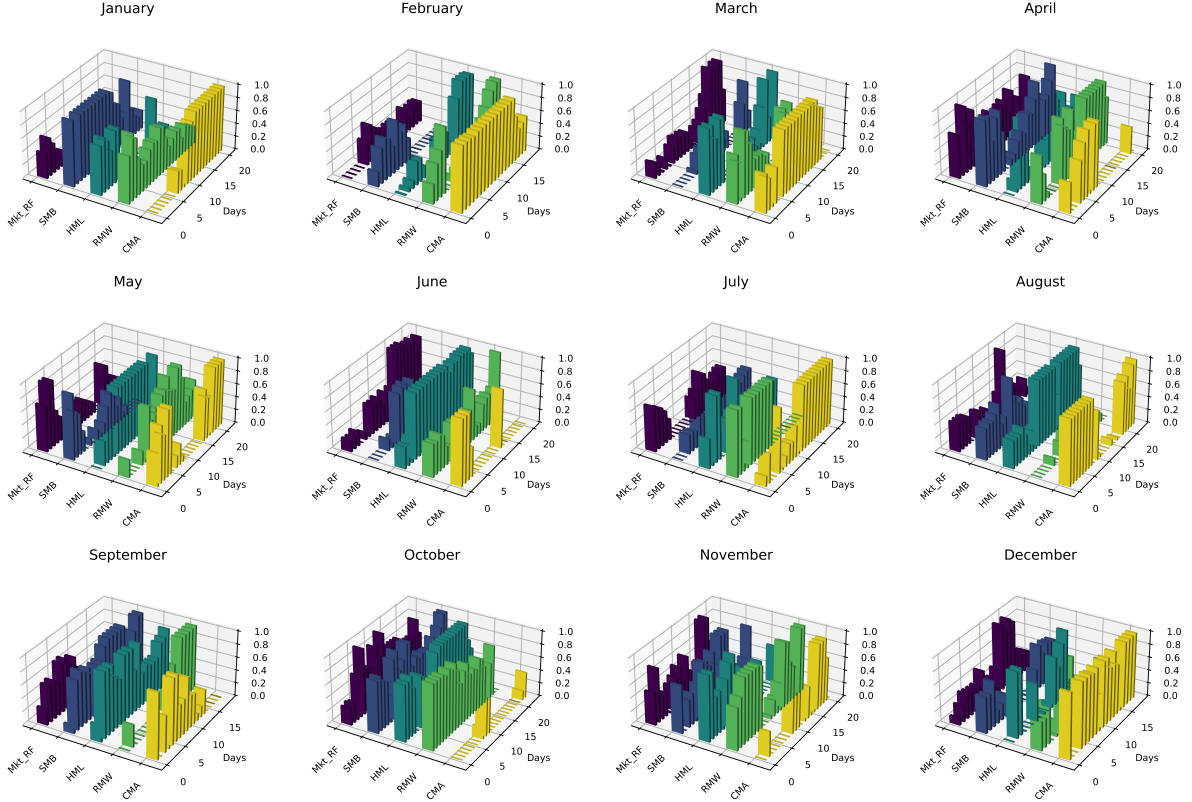}
    \caption{Variable importance of Fama-French 5 factors on 2018 returns.}
    \label{fig:fama}
\end{figure*}

\begin{table}[ht]
\centering

\resizebox{\columnwidth}{!}{%
\begin{tabular}{lccc}
\toprule
Model Variant & Coef. & p-value & $R^2$ \\
\midrule
Llama2-7b-chat & \textbf{0.150}$^*$ & \textbf{0.030} & \textbf{0.083} \\
w/o Technical Relevance & 0.111 & 0.120 & 0.067 \\
w/o Supply Chain & 0.084 & 0.261 & 0.052 \\
w/o Shared Leadership & 0.095 & 0.058 & 0.071 \\
\bottomrule
\end{tabular}%

}
\caption{Impact of Removing Graph Relationship Types on Prediction}
\label{tab:rela}
\end{table}


\subsubsection{RAG and ICL}
To effectively analyze financial events and their market impact, we employ an ICL baseline. This method provides the model with a concrete example, demonstrating the expected input format, analysis process, and output structure. By presenting a sample scenario and its corresponding analysis, we establish a clear framework for the model to follow. For the RAG method, we use text-embedding-ada-002 as our embedding model, with the same prompt template as used in ICL. The following prompt illustrates this few-shot learning technique:

\noindent
\textbf{Instruction:}
\begin{lstlisting}
You are a financial event analyst focused on analyzing the potential impacts of news reports on the market. Based on the given news content and current market structure, evaluate and output the affected companies (TICKER in SP500), the type of impact (positive, negative, or neutral), and a score representing the strength of the impact (ranging from -10 to +10, where -10 indicates a very negative impact, and +10 indicates a very positive impact). Provide specific company names and event descriptions for clarity and utility. Here is an example.
\end{lstlisting}

\noindent
\textbf{Input Example:}
\begin{lstlisting}
"Company A announces a partnership with Company B to jointly develop new technology, expected to significantly enhance production efficiency and increase market share."    
\end{lstlisting}

\noindent
\textbf{Output Format Example:}
\begin{lstlisting}
{
  "impact_analysis": {
    "affected_companies": [
      {
        "name": "Company A",
        "impact_type": "positive",
        "impact_score": 8
      },
      {
        "name": "Company B",
        "impact_type": "positive",
        "impact_score": 6
      }
    ],
    "analysis": "The partnership between Company A and Company B is 
    expected to enhance their technological capabilities and market 
    competitiveness, likely increasing their revenues and stock prices. 
  }
}    
\end{lstlisting}

\noindent
\textbf{Input (you need to analyze):}
\begin{lstlisting}
"Company A announces a partnership with Company B to jointly develop new technology, expected to significantly enhance production efficiency and increase market share."    
\end{lstlisting}

\noindent
\textbf{Knowledge Graph (current market structure you can refer to):}
\begin{lstlisting}
(Company A, Company B, supplier)
(Company C, Company D, subsidiary)
(Company E, Company F, competitor)
(Company G, Company H, partner)
(Company I, Company J, investor)
(Company Q, Company R, technology provider) ...    
\end{lstlisting}

\noindent
\textbf{Provide your result, strictly following the output format in the example, without any additional output.}

\subsection{Statistical Metrics}

This subsection introduces key statistical metrics used to evaluate the explanatory power of independent variables on the dependent variable, including Coefficient (Coef.), p-value, Coefficient of Determination (\( R^2 \)), ANOVA F-statistic (ANOVA-F), ANOVA p-value (ANOVA-p), and Effect Size (\( \eta^2 \)).

\paragraph{Coefficient (Coef.)}
The coefficient (\( \beta_i \)) represents the estimated effect of an independent variable \( X_i \) on the dependent variable \( Y \), holding all other variables constant. The regression equation is given by \( Y = \beta_0 + \beta_1 X_1 + \beta_2 X_2 + \cdots + \beta_n X_n + \epsilon \), where \( \epsilon \) is the error term.

\paragraph{p-value}
The p-value indicates the statistical significance of each coefficient, measuring the probability of observing the estimated effect under the null hypothesis that the coefficient is zero. A smaller p-value suggests stronger evidence against the null hypothesis.

\paragraph{Coefficient of Determination (\( R^2 \))}
The Coefficient of Determination (\( R^2 \)) measures the proportion of variance in the dependent variable that is explained by the independent variables. It is calculated as \( R^2 = 1 - \frac{\sum_{i=1}^{n} (y_i - \hat{y}_i)^2}{\sum_{i=1}^{n} (y_i - \bar{y})^2} \), where \( y_i \) is the observed value, \( \hat{y}_i \) is the predicted value, and \( \bar{y} \) is the mean of the observed values.

\paragraph{ANOVA F-statistic (ANOVA-F)}
The ANOVA F-statistic tests whether the regression model explains a significant proportion of variance in the dependent variable compared to a model with no predictors. It is calculated as \( F = \frac{\text{MS}_{\text{regression}}}{\text{MS}_{\text{residual}}} \), where \( \text{MS}_{\text{regression}} \) is the mean square due to regression, and \( \text{MS}_{\text{residual}} \) is the mean square due to residual error. Higher values of \( F \) suggest a better fit forr the model.

\paragraph{ANOVA p-value (ANOVA-p)}
The ANOVA p-value indicates the statistical significance of the F-statistic, reflecting the probability of obtaining the computed F-statistic under the null hypothesis that the regression model has no explanatory power.

\paragraph{Effect Size (\( \eta^2 \))}
Effect Size (\( \eta^2 \)) represents the proportion of the total variance in the dependent variable that is attributable to an independent variable or a set of independent variables. It is calculated as \( \eta^2 = \frac{\text{SS}_{\text{between}}}{\text{SS}_{\text{total}}} \), where \( \text{SS}_{\text{between}} \) is the sum of squares between groups, and \( \text{SS}_{\text{total}} \) is the total sum of squares. This metric helps determine the magnitude of the effect of the independent variables.
\subsection{Portfolio Management}

Portfolio management involves the selection and optimization of asset allocation to maximize the return within a given investment process (Hu and Lin, 2019). In this section, we describe the implementation details of five benchmark portfolio strategies: Equal Weighting, Volatility Weighting, Markowitz Model, Min-Variance Weighting, and FinRipple. These benchmarks are evaluated using metrics such as Daily Return (\( R_d \)), Sharpe Ratio (\( S_a \)), Maximum Drawdown (MDD), and Win Rate. In our experiments, we use historical data from the past 30 days as input. To simplify the comparison and ensure fairness, tax rates are set to zero across all scenarios. 

\subsubsection{Equal Weighting}
The Equal Weighting strategy assigns an equal weight to each asset in the portfolio:
\[
w_i = \frac{1}{N}, \quad i = 1, 2, \ldots, N
\]
where \( w_i \) represents the weight of asset \( i \), and \( N \) is the total number of assets.

\subsubsection{Volatility Weighting}
The Volatility Weighting strategy allocates weights inversely proportional to the historical volatility of each asset:

\begin{equation}
w_i = \frac{\frac{1}{\sigma_i}}{\sum_{j=1}^{N} \frac{1}{\sigma_j}}, \quad i = 1, 2, \ldots, N
\label{eq_a5}
\end{equation}

where \( \sigma_i \) is the historical volatility (standard deviation) of asset \( i \).

\subsubsection{Markowitz Model}
The Markowitz Model, also known as the Mean-Variance Optimization Model, aims to maximize expected return for a given level of risk or minimize risk for a given expected return:
\[
\begin{aligned}
    &\max_{\mathbf{w}} \quad \mathbf{w}^T \mathbf{\mu} - \frac{\lambda}{2} \mathbf{w}^T \mathbf{\Sigma} \mathbf{w} \\
    &\text{s.t.} \quad \mathbf{1}^T \mathbf{w} = 1, \quad \mathbf{w} \geq 0
\end{aligned}
\]
where \( \mathbf{w} \) is the vector of portfolio weights, \( \mathbf{\mu} \) is the expected return vector, \( \mathbf{\Sigma} \) is the covariance matrix of asset returns, and \( \lambda = 1 \) is the risk aversion parameter, representing a moderate balance between risk and return.

\subsubsection{Min-Variance Weighting}
The Min-Variance Weighting strategy seeks to construct a portfolio with the lowest overall risk:
\begin{align*}
    &\min_{\mathbf{w}} \quad \mathbf{w}^T \mathbf{\Sigma} \mathbf{w} \\
    &\text{s.t.} \quad \mathbf{1}^T \mathbf{w} = 1, \quad \mathbf{w} \geq 0
\label{eq_a7}
\end{align*}
where \( \mathbf{\Sigma} \) is the covariance matrix of asset returns.

\subsection{Metrics of Portofolio Management}

The benchmarks are evaluated using the following metrics:

\paragraph{Daily Return (\(R_d\))}
The daily return measures the return of an asset over one day, calculated as \( R_d = \frac{P_t - P_{t-1}}{P_{t-1}} \), where \( P_t \) is the asset price at time \( t \), and \( P_{t-1} \) is the price on the previous trading day.

\paragraph{Sharpe Ratio (\(S_a\))}
The Sharpe ratio measures investment performance compared to a risk-free asset, adjusted for risk, using the formula \( S_a = \frac{\bar{R}_a - R_f}{\sigma_a} \), where \( \bar{R}_a \) is the average annual return, \( R_f \) is the risk-free rate, and \( \sigma_a \) is the standard deviation of the return.

\paragraph{Maximum Drawdown (MDD)}
Maximum Drawdown represents the maximum observed loss from a peak to a trough of an asset's price, given by \( \text{MDD} = \max_{t \in [1,T]} \left( \frac{\max_{j \in [1,t]} P_j - P_t}{\max_{j \in [1,t]} P_j} \right) \), where \( P_t \) is the price at time \( t \), and \( T \) is the total time period considered.

\paragraph{Win Rate (Wr)}
Win Rate represents the percentage of time periods in which the portfolio achieves a positive return, defined as \( \text{Wr} = \frac{\sum_{t=1}^T \mathbb{I}(R_t > 0)}{T} \times 100\% \), where \( R_t \) is the return at time \( t \), \( T \) is the total number of time periods considered, and \( \mathbb{I}(R_t > 0) \) is an indicator function that equals 1 if \( R_t > 0 \), and 0 otherwise.

\section{Reproducibility Statement}
\subsection{Hyperparameter Selection}
We conducted hyperparameter tuning on a small-scale dataset to determine the optimal settings for minimizing the refusal-to-answer rate. The resulting hyperparameter settings are shown in Table~\ref{tab:123}, aiming to reduce the likelihood of model refusal while maintaining high response quality. In the reward function, \( \lambda \) is set to 0.1. We used LoRA (Low-Rank Adaptation) \citep{hu2021lora} to fine-tune the model, with key settings including \( \text{lora\_alpha} = 16 \), \( \text{lora\_dropout} = 0.1 \), and rank \( r = 64 \). 

\begin{table}[h]
    \centering
    \resizebox{\columnwidth}{!}{%
    \begin{tabular}{lccc}
        \toprule
        \textbf{Model} & \textbf{Temperature} & \textbf{Top-k} & \textbf{Top-p} \\
        \midrule
        Llama2-7b-chat        & 0.8 & 40 & 0.85 \\
        Llama2-13b-chat       & 0.7 & 50 & 0.90 \\
        Llama3-8b-instruct    & 0.7 & 30 & 0.80 \\
        vicuna-7b-chat        & 0.8 & 45 & 0.88 \\
        vicuna-13b-chat       & 0.7 & 50 & 0.92 \\
        Phi-3.5-mini-instruct & 0.9 & 35 & 0.86 \\
        Gemma-2-9b-it         & 0.9 & 25 & 0.83 \\
        GPT 3.5               & 0.8 & -& 0.80 \\
        GPT 4.0-preview       & 0.8 & -& 0.85 \\
        GPT 4o-mini           & 0.7 & -& 0.87 \\
        \bottomrule
    \end{tabular} 
    }
    \caption{Hyperparameter experiments.}
    \label{tab:123}
\end{table}

\subsection{Computational Resources and Code Availability}
The training and inference results required a total of over 9000 GPU hours using 25 A800 (80G) GPUs. We will release a user-friendly training framework along with the complete benchmark dataset in the future.

\end{document}